\documentclass[11pt, onecolumn,letterpaper,journal,]{IEEEtran}

\usepackage{amsfonts,amsmath,amssymb,amsthm}
\usepackage{upgreek}
\usepackage{mathtools, xparse}
\usepackage{graphicx}
\usepackage{epstopdf}
\usepackage{bm}
\usepackage{scalerel}
\usepackage{verbatim}
\usepackage{url}
\usepackage{comment}
\usepackage{epsf,pgf}
\usepackage{diagbox}
\usepackage{makecell}
\usepackage{color}
\usepackage{braket}
\usepackage{float}
\usepackage{doi}
\usepackage{qcircuit}
\usepackage{multirow}
\usepackage{listings}
\usepackage{array}
\usepackage{makecell}
\usepackage[switch]{lineno}
\usepackage{subcaption,cleveref}
\usepackage{enumitem}

\newcommand{\op}{\oplus}
\newcommand{\pr}{\mathcal{P}}
\newcommand{\xv}{\mathbf{x}}
\newcommand{\yv}{\mathbf{y}}
\newcommand{\var}[1]{\mathbf{#1}}

\newtheorem{theorem}{Theorem}[section]
\newtheorem{proposition}[theorem]{Proposition}
\newtheorem{corollary}[theorem]{Corollary}
\newtheorem{lemma}[theorem]{Lemma}
\newtheorem{definition}[theorem]{Definition}
\newtheorem{remark}[theorem]{Remark}

\def \F {{\mathbb F}}

\newcommand{\wh}{Walsh-Hadamard transform }
\newcommand{\nh}{nega-Hadamard transform }

\begin{document}
\title{Extending Forrelation: Quantum Algorithms Related to Generalized Fourier-Correlation}

\author{\IEEEauthorblockN{{\Large Suman Dutta, Subhamoy Maitra}}\\
\vskip.15cm
\IEEEauthorblockA{Applied Statistics Unit, Indian Statistical Institute,\\
Kolkata 700108, India;
sumand.iiserb@gmail.com, subho@isical.ac.in\\\vskip.3cm}
\and
\IEEEauthorblockN{{\Large Pantelimon St\u anic\u a}}\\
\vskip.15cm
\IEEEauthorblockA{Applied Mathematics Department, Naval Postgraduate School, \\Monterey, CA 93943, USA;
pstanica@nps.edu}}

\maketitle
%\date{\today}

 \thispagestyle{empty}
\begin{abstract}
%In this paper, we introduce a novel technique for sampling any arbitrary powers of the Walsh-Hadamard transform using Forrelation, extending the square sampling method of the Walsh-Hadamard transform proposed by Dutta et al. (2021).
In this paper, we study different cryptographically significant spectra of Boolean functions, including the Walsh-Hadamard, cross-correlation, and autocorrelation. The $2^k$-variation by St\u anic\u a [IEEE-IT 2016] is considered here with the formulation for any $m \in \mathbb{N}$. Given this, we present the most generalized version of the Deutsch-Jozsa algorithm, which extends the standard and previously extended versions, thereby encompassing them as special cases. Additionally, we generalize the Forrelation formulation by introducing the $m$-Forrelation and propose various quantum algorithms towards its estimation. In this regard, we explore different strategies in sampling these newly defined spectra using the proposed $m$-Forrelation algorithms and present a comparison of their corresponding success probabilities. Finally, we address the problem related to affine transformations of generalized bent functions and discuss quantum algorithms in identifying the shifts between two bent (or negabent) functions with certain modifications, and compare these with existing results.
%, identifying an error, in the process.
\end{abstract}

\section{Introduction}
\label{sec:intro}
Quantum computing is a fundamental aspect of quantum mechanics with numerous exciting possibilities. However, demonstrating the separation between the classical and quantum paradigms remains a challenging and intriguing problem. In this context, the black-box model is widely employed to illustrate the distinction in their respective domains. Some of the most well-known quantum algorithms, such as Shor's factoring~\cite{shor}, Grover's search~\cite{grover}, and Simon's hidden shift~\cite{simon}, are all defined within this black-box paradigm. Naturally, this black-box model, also known as the query model, is one of the most studied areas in quantum computer science. In this domain, problems can be studied in the exact quantum and bounded error quantum models as well as the probabilistic and deterministic classical models.

One of the simple yet fundamental algorithm in this domain is the Deutsch-Jozsa algorithm~\cite{dj}. Informally speaking, given access to a Boolean function $f$, it produces a superposition of the normalized Walsh-Hadamard transform values of $f$ as the amplitude of its pre-measurement state at the respective points. If $f$ is promised to be either constant or balanced, the Deutsch-Jozsa algorithm deterministically concludes which one it is with a single quantum query. In contrast, any deterministic classical algorithm requires exponentially many queries, which is asymptotically the maximum possible separation between these two models. However, obtaining the separation between the classical probabilistic model and the bounded error quantum model is not as straightforward and was recently resolved in~\cite{bansal}. In this direction, the `Forrelation' algorithm is of great interest.

Given two Boolean functions $f$ and $g$, the Forrelation (or Fourier-correlation) estimates the correlation between the truth table (output) of $f$ and the Walsh-Hadamard transform (at the corresponding points) of $g$~\cite{aaron}. In 2015, Aaronson et al.~\cite{forr} demonstrated that, given any two Boolean functions, the absolute value of the Forrelation being less than $1/100$ or greater than $3/5$ can be resolved using a constant number of queries in the bounded error quantum model. In contrast, it takes an exponential number of queries in the probabilistic classical model, highlighting the classical-quantum separation.

In a recent study~\cite{amc}, the Forrelation algorithm was exploited to enhance the sampling probability for different cryptographically significant spectra of Boolean functions, including the Walsh-Hadamard spectrum, cross-correlation, and auto-correlation spectra. In a subsequent work~\cite{dcc}, similar results were established involving the nega-Hadamard transforms. Additionally,~\cite{dcc} also explored various interesting properties related to the shift behavior of bent and negabent functions using Forrelation and nega-Forrelation. In a related work, St\u anic\u a~\cite{2kbent} generalized the Walsh-Hadamard transform to the $2^k$-Hadamard transform by incorporating the weight factor of the $2^k$-th primitive root of unity. This work also extended the notions of cross-correlation and auto-correlation within this framework and generalized the bent criteria for Boolean functions to $2^k$-bent, introducing the concepts of strong and weak bent functions in this context. We point to also~\cite{Tang16} and~\cite{WZ21} for further descriptions of the $2^k$-bent concept and applications to relative difference set theory.

In our present endeavor, we first generalize these fundamental concepts from their $2^k$-variation to a more comprehensive formulation, by incorporating the weight factor of the $m$-th primitive root of unity, for any $m\in\mathbb{N}$. Moreover, we generalize the Deutsch-Jozsa algorithm beyond its extended variant~\cite{exdj}, identifying the standard Deutsch-Jozsa~\cite{dj} and the extended Deutsch-Jozsa~\cite{exdj} as its specific instances. Furthermore, we extend the Forrelation formulation to $m$-Forrelation and introduce novel quantum algorithms for estimating these newly defined generalized spectra. The organization of the paper and its section-wise contributions are summarized bellow.
\subsection{Organization and Contribution}
In Section~\ref{sec:pre}, we present the preliminary results related to Boolean functions, along with the relevant quantum algorithms such as the Deutsch-Jozsa, Forrelation, nega-Forrelation, and their respective variations and applications.

Section~\ref{sec:cont1} is our first contributory section, where we focus on generalizing the existing frameworks in terms of Boolean function properties. We begin with introducing new unitaries, exploring their implications, and establishing connections to existing ones.
Then, we extend the formulation of various fundamental concepts including the Walsh-Hadamard, cross-correlation, and autocorrelation spectra, to a generalized variation for any $m\in\mathbb{N}$. In the process, we identify a previously unexamined class of real Hadamard transforms that lies between the Walsh-Hadamard and nega-Hadamard transformations, addressing a gap in the literature. Additionally, we introduce the most generalized version of the Deutsch-Jozsa algorithm, which extends both the standard Deutsch-Jozsa and its prior extended version, incorporating them as special cases. Furthermore, we extend the Forrelation formulation to $m$-Forrelation and propose new quantum algorithms for estimating them for a given set of Boolean functions.

In Section~\ref{sec:cont2}, we present various sampling strategies of these newly defined spectra of Boolean functions using the generalized Forrelation algorithms, and present the comparison graph based on the corresponding sampling probabilities.

Section~\ref{sec:cont3} presents our final contribution, where we study affine transformations of generalized bent functions. We refute a prior claim that two negabent functions cannot exhibit a hidden shift by providing a counterexample and extend hidden shift finding algorithms, originally developed for bent functions, to more general affine transformations and to negabent functions.

Section~\ref{sec:con} concludes the paper with a brief summary of our major contributions and outlines potential future research directions in this area.

\section{Preliminaries}
\label{sec:pre}
In this section, we present the necessary definitions and results required to proceed with this paper. For more on Boolean functions and their cryptographic properties, the reader can consult~\cite{CarletBook21,stanbook}.

Let $\mathbb{F}_2=\{0,1\}$ be the prime field of characteristic $2$ and $\mathbb{F}^n_2\equiv \{\var{x}=\left(x_1,x_2,\ldots, x_n \right):x_i\in\mathbb{F}_2,1\leq i \leq n \}$ (usual operations) be the vector space of dimension $n$ over $\mathbb{F}_2$. Throughout the paper, elements of $\mathbb{F}_2^n$ are written in bold to distinguish them from scalar elements of $\mathbb{F}_2$. We define a \textit{Boolean function} $f$ as a mapping $f:\mathbb{F}^n_2 \rightarrow \mathbb{F}_2.$
%We call a function from $\F_2^{n}$ to ${\mathbb Z}_q$ ($q \geq 2 $) a {\em generalized}   {\em Boolean function} on $n$ variables, and denote the set of all generalized Boolean functions by $\cGB_{n}^q$ and, when $q=2$, by $\cB_{n}$.
The set of all $n$-variable Boolean functions is denoted by $\mathcal{B}_n$. 
The \textit{(Hamming) weight} of a binary string $\bm{\omega}=(\omega_1,\ldots ,\omega_n)\in\mathbb{F}_2^n$ is given by the number of $1$'s present in the bit pattern of $\bm{\omega}$, that is, $wt(\bm{\omega})=\sum_{i=1}^n\omega_i$.
The Hamming weight of a Boolean function is the Hamming weight of its output vector, that is, {\em truth table} (under some ordering of the input).

Given $f\in\mathcal{B}_n$, the \textit{Walsh-Hadamard transform} (a discrete version of the Fourier transform) of $f$ at a point $\bm{\omega}=(\omega_1,\ldots ,\omega_n)\in\mathbb{F}_2^n$ is a real valued function, $W_f:\mathbb{F}_2^n\rightarrow \left[-2^{n/2},2^{n/2} \right]\subset\mathbb{R}$, defined as 
$W_f(\bm{\omega})=2^{-n/2}\sum_{\xv\in\mathbb{F}_2^n}(-1)^{f(\xv)\oplus \xv\cdot\bm{\omega}}$.  \textit{Parseval's identity}  constrains the Walsh transform values by $\sum_{\bm{\omega}\in\mathbb{F}_2^n}W^2_f(\bm{\omega})=2^n$. If the Walsh transform values of a Boolean function $f\in\mathcal{B}_n$ are equally distributed over all ($2^n$-many) points, i.e., if $W_f(\bm{\omega})=\pm 1$ for all $\bm{\omega}\in\mathbb{F}_2^n$, then $f$ is called a \textit{bent function}. The bent functions are defined only for even $n\in \mathbb{N}$ and have the highest nonlinearity. Let $f,g\in\mathcal{B}_n$ are bent such that $(-1)^{g(\xv)}=W_f(\xv)$ for all $\xv\in\mathbb{F}_2^n$, then $f$ and $g$ are called \textit{dual}, denoted by $g=\widehat{f}$. Note that, $g=\widehat{f}$ implies $f=\widehat{g}$, and therefore, $\widehat{\widehat{f}}=\widehat{g} = f$. Moreover, A bent function $f$ is self dual if $f=\widehat{f}$.

Similarly, given $f\in\mathcal{B}_n$, the \textit{nega-Hadamard transform}~\cite{nega} of $f$ at $\bm{\omega}\in\mathbb{F}_2^n$ is a complex valued function, $N_f:\mathbb{F}_2^n\rightarrow \mathbb{C}$, defined as $N_f(\bm{\omega}) = 2^{-n/2}\sum_{\xv\in\mathbb{F}_2^n}(-1)^{f(\xv)\oplus \xv\cdot\bm{\omega}} i^{wt(\xv)}$, where the constraint, $\sum_{\bm{\omega}\in\mathbb{F}_2^n}|N_f(\bm{\omega})|^2$ $= \sum_{\bm{\omega}\in\mathbb{F}_2^n}N_f(\bm{\omega})\overline{N_f(\bm{\omega})}=2^n$ is known as the \textit{nega-Parseval identity}~\cite{pre}. For even $n$, $f\in\mathcal{B}_n$ is called \textit{negabent} if $|N_f(\bm{\omega})|=1$ for all $\bm{\omega}\in\mathbb{F}_2^n$, where $|z|$ denotes the modulus of a complex number $z$. Furthermore, $f\in\mathcal{B}_n$ (for even $n$) is called bent-negabent if $f$ is both bent and negabent.

In~\cite{2kbent}, the discrete Fourier transforms of $f\in\mathcal{B}_n$ have been generalized to $2^k$-Hadamard transform, a complex-valued function ($\mathcal{H}_f^{(2^k)}:\mathbb{F}_2^n\rightarrow \mathbb{C}$), defined as 
$\mathcal{H}_f^{(2^k)}\left( \bm{\omega}\right) = 2^{-n/2}\sum_{\xv\in\mathbb{F}_2^n}(-1)^{f(\xv)\oplus \xv\cdot\bm{\omega}}\zeta_{2^k}^{wt(\xv)}$
at $\bm{\omega}\in\mathbb{F}_2^n$,
where $\zeta_{2^k}=e^{2\pi i/2^k}$ is a primitive $2^k$-complex root of $1$. The \textit{$2^k$-Parseval identity} is given by $\sum_{\bm{\omega}\in\mathbb{F}_2^n}|\mathcal{H}^{(2^k)}_f(\bm{\omega})|^2 = \sum_{\bm{\omega}\in\mathbb{F}_2^n} \mathcal{H}^{(2^k)}_f(\bm{\omega})\overline{\mathcal{H}^{(2^k)}_f(\bm{\omega})}=2^n$. A Boolean function $f\in\mathcal{B}_n$ satisfying $|\mathcal{H}_f^{(2^k)}(\bm{\omega})|=1$ for all $\bm{\omega}\in\mathbb{F}_2^n$ is called a $2^k$-bent function.
\begin{remark}
It is important to note that for $k=0$, $\zeta_{2^0}=e^{2\pi i}=1$, resulting in $\mathcal{H}_f^{(2^k)}=W_f$, the Walsh-Hadamard transform. Similarly, for $k=2$, $\zeta_{2^k}=i$, which leads to $\mathcal{H}_f^{(2^k)}=N_f$, the nega-Hadamard transform.
	
Additionally, we identify the existence of a missing $2^k$-Hadamard transform for $k=1$ $(\zeta_{2^1}=-1)$, which is also a real-valued function that lies between the Walsh-Hadamard and nega-Hadamard transforms. For $f\in\mathcal{B}_n$, it is described at a point $\bm{\omega}\in \mathbb{F}_2^n$ as follows:
\[
\mathcal{H}_f^{(2^1)}\left( \bm{\omega}\right) = 2^{-n/2}\sum_{\xv\in\mathbb{F}_2^n}(-1)^{f(\xv)\oplus \xv\cdot\bm{\omega}}(-1)^{wt(\xv)}.
    \]
\end{remark}

%Moreover, a Boolean function $f\in\mathcal{B}_n$ is a strong $2^k$-bent function if and only if $f$ is $2^l$-bent for all $l\leq k$. Also, a function $f\in\mathcal{B}_n$ is weak $2^{k}$-bent if and only if $f\oplus s_{2^k-1}$ is a strong $2^{k-1}$-bent function, where $s_t(\xv)=\bigoplus_{1\leq i_1<\ldots <i_t\leq n}x_{i_1}\cdots x_{i_t}$ is the (modulo $2$) elementary symmetric polynomial of degree $t$.\\

Given $f,g\in\mathcal{B}_n$, the \textit{crosscorrelation} of $f$ and $g$ at $\var{y}\in\mathbb{F}_2^n$ is defined as $C_{f,g}(\var{y})=\sum_{\xv\in\mathbb{F}_2^n}{(-1)^{f(\xv)\oplus g(\xv\oplus \var{y})}}$, where taking $f=g$ we obtain the \textit{autocorrelation} of $f\in\mathcal{B}_n$, represented by $C_f(\var{y})=\sum_{\var{x}\in\mathbb{F}_2^n}{(-1)^{f(\xv)\op f(\xv\oplus \var{y})}}$. From~\cite{pre}, it is known that a Boolean function $f\in\mathcal{B}_n$ is bent if and only if $C_f(\var{y})=0$ for all $\var{y}\in\mathbb{F}_2^n \setminus\{0^n\}$.
In the same vein, the \textit{nega-crosscorrelation} of $f,g\in\mathcal{B}_n$ at $\var{y} \in \mathbb{F}_2^n$ is defined as $\displaystyle \widehat{C}_{f,g}(\var{y})=\sum_{\xv \in \mathbb{F}_2^n} (-1)^{f(\xv)\oplus g(\xv \oplus \var{y})} (-1)^{\var{x}\cdot \var{y}}$, and the \textit{nega-autocorrelation} of $f\in\mathcal{B}_n$ as $\displaystyle \widehat{C}_{f}(\var{y})=\sum_{\xv \in \mathbb{F}_2^n} (-1)^{f(\xv)\oplus f(\xv \oplus \var{y})}(-1)^{\var{x}\cdot \var{y}}$.
In~\cite{2kbent}, the crosscorrelation and the autocorrelation have been generalized to \textit{$2^k$-crosscorrelation} and \textit{$2^k$-autocorrelation}, respectively, as follows.
$$C_{f,g}^{(2^k)}(\var{y}) = \sum_{\xv\in\mathbb{F}_2^n}(-1)^{f(\xv)\oplus g(\xv\oplus \var{y})}\left(\zeta_{2^k}^2\right)^{\xv\odot\var{y}}, \text{ and } C_{f}^{(2^k)}(\var{y}) = \sum_{\xv\in\mathbb{F}_2^n}(-1)^{f(\xv)\oplus f(\xv\oplus \var{y})}\left(\zeta_{2^k}^2\right)^{\xv\odot\var{y}},$$
where $\xv\odot\var{y}$ denotes the inner product in $\mathbb{C}\times \mathbb{C}$. 
\begin{remark}
Note that for both $k=0$ and $k=1$, the $2^k$-crosscorrelation and $2^k$-autocorrelation reduce to the standard crosscorrelation and autocorrelation, respectively. Similarly, for $k=2$, the $2^k$-crosscorrelation and $2^k$-autocorrelation correspond to the nega-crosscorrelation and nega-autocorrelation, respectively.
	
Since the expression for $2^k$-crosscorrelation (autocorrelation) includes the term $\zeta_{2^k}^2$, the variation in $\zeta_{2^k}$, where $\zeta_{2^k}=1$ for $k=0$ and $\zeta_{2^k}=-1$ for $k=1$, does not have an overall impact. Consequently, unlike the $2^k$-Hadamard transform, there is no missing $2^k$-crosscorrelation (autocorrelation) value.
\end{remark}

Moreover, the $2^k$-crosscorrelation is related to the product of $2^k$-Hadamard transform~\cite[Theorem 3]{2kbent} as follows.
$$C_{f,g}^{(2^k)}(\var{y})=\zeta_{2^k}^{wt(\var{y})}\sum_{\var{u}\in\mathbb{F}_2^n}\mathcal{H}_{f}^{(2^k)}(\var{u})\overline{\mathcal{H}_{g}^{(2^k)}(\var{u})}(-1)^{\var{u}\cdot\var{y}}.$$
Recall that there exist similar results for crosscorrelation~\cite[Theorem 3.1]{cross} as well as for nega-crosscorrelation~\cite[Lemma 4]{pre}.
$$C_{f,g}(\var{y})=\sum_{\var{u}\in\mathbb{F}_2^n}W_{f}(\var{u})W_{g}(\var{u})(-1)^{\var{u}\cdot\var{y}},\quad \widehat{C}_{f,g}(\var{y})=i^{wt(\var{y})}\sum_{\var{u}\in\mathbb{F}_2^n}N_{f}(\var{u})\overline{N_{g}(\var{u})}(-1)^{\var{u}\cdot\var{y}}.$$
In Section~\ref{sec:cont1}, we generalize the $2^k$-Hadamard transform, $2^k$-crosscorrelation, $2^k$-autocorrelation and the related results for any $m\in\mathbb{N}$, establishing the existing results as specific cases within this generalized framework.

In the \textit{black-box model} of quantum algorithms, we are given the oracle access $\left(U_f\right)$ of an unknown $f\in\mathcal{B}_n$, and the goal is to determine certain specific properties of $f$ with the minimum number of queries to $U_f$. Given $\xv\in \mathbb{F}_2^n$ and $a\in\mathbb{F}_2$, the functioning of $U_f$ on an $(n+1)$ qubit state $\ket{\xv}\ket{a}$ is defined as $U_f\ket{\xv}\ket{a}=\ket{\xv}\ket{a \oplus f(\xv)}$. 
For $\ket{a}=\ket{-}$, we have $U_f\ket{\xv}\ket{-}= (-1)^{f(\xv)}\ket{\xv}\ket{-}$, which is also known as the \textit{phase-kickback}.

At this point, let us briefly explain the Deutsch-Jozsa algorithm~\cite{dj}.
Given oracle access to an unknown $f\in\mathcal{B}_n$, with the promise that $f$ is either a constant or a balanced Boolean function, the Deutsch-Jozsa algorithm can deterministically identify which one it is by producing a superposition state where the amplitude of an individual state is given by the normalized \wh of the function at that point
$$2^{-n}\sum_{\xv, \var{y} \in \mathbb{F}_2^n}(-1)^{f(\var{x})\oplus\var{x}\cdot\var{y}}\ket{\var{y}}=2^{-n/2} \sum_{\var{y} \in \mathbb{F}_2^n} W_f(\var{y}) \ket{\var{y}}.$$
Essentially, for a constant Boolean function, the complete Walsh spectrum is supported over the all-zero point, makes it certain to observe the state $\ket{0^n}$, upon measurement. However, for a balanced Boolean function, the Walsh transform at the all-zero point is $0$, thus the state $\ket{0^n}$ never appears when measured. Therefore, the presence or absence of the all-zero state in the measurement results concludes with certainty whether $f$ is constant or balanced, using only a single query to the oracle $U_f$. A schematic diagram of the corresponding quantum circuit is shown in Fig.~\ref{fig:dj}.
\begin{figure}[ht]
	\centering
	\includegraphics[scale=1.2]{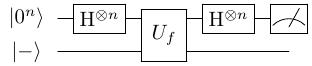}
	\caption{Quantum circuit for Deutsch-Jozsa algorithm~\cite{dj}.}
	\label{fig:dj}
\end{figure}

In 2017, a modified version of this algorithm, termed the extended Deutsch-Jozsa algorithm~\cite{exdj} was proposed in~\cite{exdj}. There, instead of using all Hadamard gates before measurement, as in~\cite{dj}, a combination of Hadamard and nega-Hadamard gates subject to a bit pattern $\var{c}\in\mathbb{F}_2^n$ was applied.
Observe that, for $\var{c}=0^n$, it becomes the standard Deutsch-Jozsa algorithm, producing the superposition state with the normalized \wh value of the respective states as the individual amplitudes. On the other extreme, when $\var{c}=1^n$, only the nega-Hadamard gates are applied. As a result, the algorithm produces the superposition where the amplitude of the individual states are given by the normalized \nh of the respective states,
$$2^{-n}\sum_{\xv, \var{y} \in \mathbb{F}_2^n}(-1)^{f(\var{x})\oplus\var{x}\cdot\var{y}}(i)^{wt(\var{x})}\ket{\var{y}} = 2^{-n/2} \sum_{\var{y} \in \mathbb{F}_2^n} N_f(\var{y}) \ket{\var{y}}.$$
A schematic diagram of the quantum circuit corresponding to $\var{c}=1^n$ is shown in Fig.~\ref{fig:exdj}.
\begin{figure}[ht]
\centering
\includegraphics[scale=1.2]{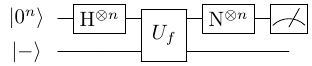}
\caption{Quantum circuit for extended Deutsch-Jozsa algorithm with $\var{c}=1^n$~\cite{exdj}.}
\label{fig:exdj}
\end{figure}
In Section~\ref{sec:cont1}, we further generalize the Deutsch-Jozsa algorithm, and show the existing variations referred in~\cite{dj} and~\cite{exdj} as its sub-cases.\\

The Forrelation problem, introduced by Aaronson~\cite{aaron}, is one of the key results in the study of showing separation between the bounded error quantum model and the randomized classical model~\cite{forr}. The Forrelation is defined as follows.
\begin{definition}[Forrelation~\cite{aaron}]
	\label{def:forr}
	Given oracle access to $f_1,f_2\in\mathcal{B}_n$, the ($2$-fold) Forrelation measures the amount of correlation between $f_1$ and the Walsh-Hadamard transform of $f_2$, by estimating the following:
	$$\Phi_{f_1,f_2}=\frac{1}{2^{n}} \sum_{\var{x_1}\in\mathbb{F}_2^n} (-1)^{f_1(\var{x_1})}W_{f_2}(\var{x_1}) 
	= \frac{1}{2^{3n/2}} \sum_{\var{x_1},\var{x_2} \in\mathbb{F}_2^n } (-1)^{f_1(\var{x_1})}(-1)^{\var{x_1}\cdot \var{x_2}}(-1)^{f_2(\var{x_2})}.$$
	The Forrelation formulation can be further extended for $k(>2)$ Boolean functions $f_1, \ldots , f_k\in\mathcal{B}_n$ to
	$$\Phi_{f_1,\ldots , f_k} = \frac{1}{2^{\frac{(k+1)n}{2}}}\sum\limits_{\var{x_1}, \ldots , \var{x_k} \in\mathbb{F}_2^n}
	(-1)^{f_1(\var{x_1})}(-1)^{\var{x_1}\cdot \var{x_2}} (-1)^{f_2(\var{x_2})}  \ldots (-1)^{\var{x_{k-1}}\cdot \var{x_k}}
	(-1)^{f_k(\var{x_k})},$$
	known as the $k$-fold Forrelation.
\end{definition}

Given oracle access to $k$-many Boolean functions, Aaronson et al.~\cite{forr} provides two efficient quantum algorithms for estimating $\Phi_{f_1,\ldots , f_k}$, the first, by using $k$ sequential queries and the second, by using $\lceil  \frac{k}{2} \rceil$ parallel queries. Given $f_1,f_2,f_3\in\mathcal{B}_n$, the schematic diagrams of the quantum algorithms for estimating the $3$-fold Forrelation, $\Phi_{f_1, f_2, f_3}$ (one using $3$ sequential queries and another using $\lceil \frac{3}{2} \rceil=2$-parallel queries) are presented in Fig.~\ref{fig:forr}. For the step-by-step analysis of theses algorithms, the readers are referred to~\cite{amc}:
\begin{figure}[ht]
	\centering
	\begin{subfigure}[b]{0.7\textwidth}
		\includegraphics[scale=1.2]{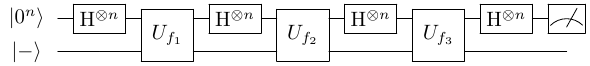}
		\caption{$3$-fold Forrelation using $3$ sequential queries.}
		\label{fig:forr-q3f3}\vspace{0.2cm}
	\end{subfigure}\\
	\begin{subfigure}[b]{0.7\textwidth}
		\includegraphics[scale=1.2]{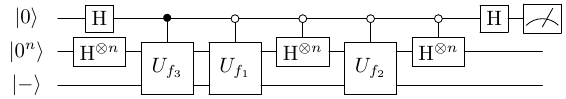}
		\caption{$3$-fold Forrelation using $2$ parallel queries.}
		\label{fig:forr-q2f3}
	\end{subfigure}
	\caption[]{Quantum algorithms for estimating the $3$-fold Forrelation~\cite{forr}.}
	\label{fig:forr}
\end{figure}

Although the Forrelation algorithm was initially used to demonstrate theoretical separation between two complexity classes, recent work~\cite{amc} has shown that it can also be leveraged for efficient sampling of different cryptographic spectra of Boolean functions, including the Walsh-Hadamard, crosscorrelation, and autocorrelation spectra. In this regard, we have the following results from~\cite{amc}.
\begin{theorem}[\cite{amc}]
	\label{thm:sample-hadamard}
	Given oracle access to $f\in\mathcal{B}_n$ and a set of points $S\subseteq\mathbb{F}_2^n$, by defining $f_1=f_3=f$ and $f_2=g\in\mathcal{B}_n$, where $g(\var{x})=1$ for $\var{x}\in S$ and $0$ otherwise, the $3$-query $3$-fold Forrelation algorithm (Fig.~\textup{\ref{fig:forr-q3f3}}) measures any state other than the all-zero state with probability of approximately $\frac{4}{2^n}\sum_{\var{x}\in S}|W_f(\var{x})|^2$.
\end{theorem}
\begin{theorem}[\cite{amc}]
	\label{thm:sample-cross}
	Given oracle access to $f,g\in\mathcal{B}_n$,  and a point $\var{y}\in\mathbb{F}_2^n$, the $3$-query $3$-fold Forrelation algorithm (Fig.~\textup{\ref{fig:forr-q3f3}}) measures the all-zero state with probability $2^{-2n} \left(C_{f,g}(\var{y})^2\right)$, by setting $f_1=f$, $f_2=\mathbb{L}_{\var{y}}$, and $f_3=g$, where $\mathbb{L}_{\var{y}}=\var{y}\cdot \var{x}$ is a linear function in $\mathcal{B}_n$.
	With an exact same setup, the $2$-query $3$-fold Forrelation algorithm (Fig.~\textup{\ref{fig:forr-q2f3}}) measures $\ket{0}$ with a probability $\frac{1}{2}\left(1+\frac{C_{f,g}(\var{y})}{2^n}\right)$.
	Additionally, by setting $f=g$, the above results apply to the autocorrelation coefficient $C_{f}(\var{y})$, at the point $\var{y}\in\mathbb{F}_2^n$.
\end{theorem}

Following the Forrelation formulation, the concept of ($3$-fold) nega-Forrelation was introduce in~\cite{dcc}, defined as follows.
\begin{definition}[nega-Forrelation~\cite{dcc}]
	\label{def:nega-forr}
	Given oracle access to $f_1,f_2,f_3\in\mathcal{B}_n$, the $3$-fold nega-Forrelation measures certain kind of correlation between the Boolean function $f_1$, the \nh of $f_2$ and the conjugate \nh of $f_3$, mathematically formulated as
    \allowdisplaybreaks
\begin{align*}
\eta_{f_1,f_2,f_3}&=\frac{1}{2^{n}}\sum_{\var{x}\in\mathbb{F}_2^n}{(-1)^{f_1(\var{x})}N_{f_2}(\var{x})\overline{N}_{f_3}(\var{x})}\\
&=\frac{1}{2^{2n}}\sum_{\var{x_1}, \var{x_2} , \var{x_3}\in\mathbb{F}_2^n} (-1)^{f_1(\var{x_1})} (-1)^{f_2(\var{x_2})\op \var{x_1}\cdot \var{x_2}} (i)^{wt(\var{x_2})} (-1)^{f_3(\var{x_3})\op \var{x_1}\cdot \var{x_3}}(-i)^{wt(\var{x_3})}.
\end{align*}
\end{definition}

In a similar manner, two efficient quantum algorithms were proposed for estimating the $3$-fold nega-Forrelation~\cite{dcc} assuming the oracle access to the Boolean functions $f_1,f_2,f_3 \in\mathcal{B}_n$, one using $3$ sequential queries and the other using $\lceil  \frac{3}{2} \rceil=2$ parallel queries. For a detailed analysis of these algorithms, the readers are referred to~\cite{dcc}. The schematic diagrams of the corresponding quantum algorithms are shown in Fig.~\textup{\ref{fig:nega}}.
\begin{figure}[ht]
	\centering
	\begin{subfigure}[b]{0.7\textwidth}
		\includegraphics[scale=1.2]{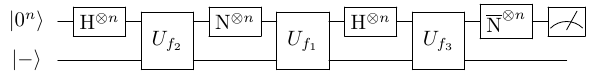}
		\caption{$3$-fold nega-Forrelation using $3$ sequential queries.}
		\label{fig:nega-q3f3}\vspace{0.2cm}
	\end{subfigure}
	\begin{subfigure}[b]{0.7\textwidth}
		\includegraphics[scale=1.2]{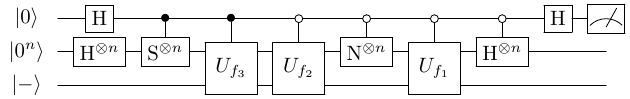}
		\caption{$3$-fold nega-Forrelation using $2$ parallel queries.}
		\label{fig:nega-q2f3}
	\end{subfigure}
	\caption[]{Quantum algorithms for estimating $3$-fold nega-Forrelation~\cite{dcc}.}
	\label{fig:nega}
\end{figure}

As established in~\cite{dcc}, the $3$-fold nega-Forrelation can also be utilized for the efficient sampling of nega-Hadamard, nega-crosscorrelation, and nega-autocorrelation spectra. In this direction, we present the following results from~\cite{dcc}.

\begin{theorem}[\cite{dcc}]
	\label{thm:sample-negahadamard}
	Given oracle access to $f\in\mathcal{B}_n$ and a set of points $S\subseteq\mathbb{F}_2^n$, by defining $f_1=f_3=f$ and $f_2=g\in\mathcal{B}_n$, where $g(\var{x})=1$ for $\var{x}\in S$ and $0$ otherwise, the $3$-query nega-Forrelation algorithm (Fig.~\textup{\ref{fig:nega-q3f3}}) measures any state other than the all-zero state with probability of approximately $\frac{4}{2^n}\sum_{\var{x}\in S}|N_f(\var{x})|^2$.
\end{theorem}

\begin{theorem}[\cite{dcc}]
	\label{thm:sample-negacross}
	Given oracle access to $f,g\in\mathcal{B}_n$,  and a point $\var{y}\in\mathbb{F}_2^n$, the $3$-query nega-Forrelation algorithm (Fig.~\textup{\ref{fig:nega-q3f3}}) measures the all-zero state with probability $2^{-2n}|\widehat{C}_{f,g}(\var{y})|^2$, by setting $f_1=f$, $f_2=\mathbb{L}_{\var{y}}$, and $f_3=g$.
	%, where $\mathbb{L}_{\var{y}}=\var{y}\cdot \var{x}$ is a linear function in $\mathcal{B}_n$. 
	With an exact same setup, the $2$-query nega-Forrelation algorithm (Fig.~\textup{\ref{fig:nega-q2f3}}) measures $\ket{0}$ with a probability $\frac{1}{2}\left(1+\Re\left(\frac{(-i)^{wt(\var{y})}\widehat{C}_{f,g}(\var{y})}{2^n}\right)\right)$, where $\widehat{C}_{f,g}(\var{y})$ is the nega-crosscorrelation of $f$ and $g$ at the point $\var{y}\in\mathbb{F}_2^n$. Additionally, by setting $f=g$, the above results hold for the nega-autocorrelation coefficient $\widehat{C}_{f}(\var{y})$, at $\var{y}\in\mathbb{F}_2^n$.
\end{theorem}

In Section~\ref{sec:cont1}, we extend the Forrelation formulation to $m$-Forrelation (Definition~\ref{def:m-forr}), which measures certain kind of correlation between the truth table of a Boolean function with the $m$-Hadamard transforms (Definition~\ref{def:m-Hadamard}) of other Boolean functions.
%

%%%%%%%%%% New Section %%%%%%%%%%%
\section{Generalization of Boolean functions' spectra}
\label{sec:cont1}
In this section, we generalize significant spectra of Boolean functions, including the Walsh-Hadamard, crosscorrelation, and autocorrelation spectra, and then present the most extended version of a Deutsch-Jozsa-like quantum algorithm, building upon the work of~\cite{dj,exdj}. Finally, we introduce the most comprehensive form of Forrelation, that encompasses both Forrelation~\cite{aaron} and nega-Forrelation~\cite{dcc} as its special cases, and provide the necessary quantum algorithm for estimating it.

In this direction, let us introduce a single qubit unitary operator, $\Omega_m$, with parameter $m\in\mathbb{N}$, defined as $\Omega_m=\frac{1}{\sqrt{2}}\begin{psmallmatrix}
	1 & \zeta_m\\
	1 & -\zeta_m
\end{psmallmatrix}$, 
where $\zeta_m=e^{2\pi i/m}$ denotes the primitive $m$-th complex root of unity, i.e., $\left(\zeta_m\right)^m = 1$. 
This gate can be derived as a special case from the well known and generalized
$U_3(\theta, \phi, \lambda) = \begin{psmallmatrix} 
\cos{\frac{\theta}{2}} & -e^{i\lambda} \sin{\frac{\theta}{2}}\\
e^{i\phi}\sin{\frac{\theta}{2}} & e^{i(\phi+\lambda)} \cos{\frac{\theta}{2}}
\end{psmallmatrix}$ structure.
The conjugate of $\zeta_m$ is given by $\overline{\zeta_m} = e^{-2\pi i/m} = \zeta_m^{-1}$. Moreover, the inverse of the unitary operator $\Omega_m$ is given by $\Omega_m^{\dagger} =  \frac{1}{\sqrt{2}}\begin{psmallmatrix}
	1 & 1\\
	\zeta_m^{-1} & -\zeta_m^{-1}
\end{psmallmatrix}$ such that $\Omega_m\Omega_m^{\dagger} = \Omega_m^{\dagger}\Omega_m = I_2$, where $I_2$ is the $2\times 2$ identity matrix $\begin{psmallmatrix}
	1 & 0\\
	0 & 1
\end{psmallmatrix}$.
The action of $\Omega_m$ on the basis states is given by:
$$\ket{0}\xrightarrow{\Omega_m} \frac{1}{\sqrt{2}}\left(\ket{0}+\ket{1}\right), \ket{1}\xrightarrow{\Omega_m} \frac{1}{\sqrt{2}}\left(\zeta_m\ket{0}-\zeta_m\ket{1}\right).$$
In general, when $n$-many $\Omega_m$ apply on an $n$-qubit state $\ket{\psi}=\sum_{\xv\in\mathbb{F}_2^n}\alpha_{\xv}\ket{\xv}$ with $\sum_{\xv\in\mathbb{F}_2^n}|\alpha_{\xv}|^2=1$, the resultant state is given by:
$$\Omega_m^{\otimes n}\left(\sum_{\xv\in\mathbb{F}_2^n}\alpha_{\xv}\ket{\xv} \right)=2^{-n/2}\sum_{\xv\in\mathbb{F}_2^n}\alpha_{\xv}\left(\sum_{\var{y}\in\mathbb{F}_2^n}(-1)^{\xv\cdot \var{y}}\zeta_m^{wt(\xv)}\ket{\var{y}}\right).$$
\begin{remark}
	Note that for $m = 1$, the newly defined unitary operator reduces to the Hadamard gate, $\Omega_1 = \mathrm{H}$. Similarly, for $m = 4$, the unitary operator $\Omega_m$ corresponds to the nega-Hadamard gate $\mathrm{N}$, as described in~\textup{\cite{dcc}}. Moreover, for $m = 2$, there exists a real unitary matrix, defined as $\Omega_2 = \frac{1}{\sqrt{2}}\begin{psmallmatrix}
		1 & -1\\ 1 & 1
	\end{psmallmatrix}$, such that
	$\Omega_2\left(\ket{0} \right) = \frac{1}{\sqrt{2}}\left(\ket{0} + \ket{1}\right)$ and $\Omega_2\left( \ket{1}\right) = \frac{1}{\sqrt{2}} \left(-\ket{0} + \ket{1}\right)$.
\end{remark}

In a similar manner, we define another single-qubit unitary by conjugating the elements of $\Omega_m$, denoted as $\overline{\Omega}_m$, which is given by $\overline{\Omega}_m = \frac{1}{\sqrt{2}}\begin{pmatrix}
	1 & \overline{\zeta_m}\\ 1 & -\overline{\zeta_m}
\end{pmatrix}$. The inverse of $\overline{\Omega}_m$ is given by $\overline{\Omega}_m^{\dagger} = \frac{1}{\sqrt{2}}\begin{pmatrix}
	1 & 1\\
	\zeta_m & -\zeta_m
\end{pmatrix}$, such that $\overline{\Omega}_m\overline{\Omega}_m^{\dagger} = \overline{\Omega}_m^{\dagger} \overline{\Omega}_m = \mathrm{I}_2$. The action of $n$-many $\overline{\Omega}_m$ on an $n$-qubit state $\ket{\psi}=\sum_{\xv\in\mathbb{F}_2^n}\alpha_{\xv}\ket{\xv}$ is given by
$$\overline{\Omega}_m^{\otimes n}\left(\sum_{\xv\in\mathbb{F}_2^n}\alpha_{\xv}\ket{\xv} \right)=2^{-n/2}\sum_{\xv\in\mathbb{F}_2^n}\alpha_{\xv}\left(\sum_{\var{y}\in\mathbb{F}_2^n}(-1)^{\xv\cdot \var{y}}(\overline{\zeta_m})^{wt(\xv)}\ket{\var{y}}\right).$$
Furthermore, for $m=4$, the newly defined $\overline{\Omega}_m$ coincides with the conjugate nega-Hadamard gate, $\overline{\mathrm{N}}$, largely utilized in~\cite{dcc}.

Additionally, we define another single-qubit gate, $\mathrm{S}_m$, which generalizes the phase gate (S), given by $\mathrm{S}_m=\begin{pmatrix}
	1 & 0\\0 & \zeta_m
\end{pmatrix}$. The inverse of $\mathrm{S}_m$ is given by $\mathrm{S}_m^{\dagger}= \begin{pmatrix}
	1 & 0\\0 & \zeta_m^{-1}
\end{pmatrix}$, and so $\mathrm{S}_m\mathrm{S}_m^{\dagger} = \mathrm{S}_m^{\dagger}\mathrm{S}_m = \mathrm{I}_2$. Notably, for $m=1$, $\mathrm{S}_m$ reduces to the identity gate. When $m=2$, we have $\zeta_2 = -1$, making $\mathrm{S}_2$ equivalent to the Pauli Z gate. Similarly, for $m=4$, $\mathrm{S}_m$ corresponds to the standard phase gate S, and for $m=8$, $\mathrm{S}_m$ is equivalent to the well-known T gate. The existing special cases of this newly defined unitary $\mathrm{S}_m$ are summarized as follows:
$$\mathrm{S}_1= \mathrm{I},\quad \mathrm{S}_2 = \mathrm{Z},\quad \mathrm{S}_4 = \mathrm{S},\quad \mathrm{S}_8 = \mathrm{T}.$$
More generally, when $n$ instances of $\mathrm{S}_m$ are applied to an $n$-qubit state $\ket{\psi} = \sum_{\xv\in\mathbb{F}_2^n} \alpha_{\xv} \ket{\xv}$, the resultant state acquires a multiplicative phase factor of $\zeta_m^{wt(\xv)}$, that is,
$$\displaystyle{\mathrm{S}_{m}^{\otimes n}\left(\sum_{\xv\in \mathbb{F}_2^n} {\alpha_{\var{x}}\ket{\xv}} \right) = \sum_{\xv\in \mathbb{F}_2^n}{\alpha_{\xv}\zeta_m^{wt(\xv)}\ket{\xv}}.}$$

Next, we generalize the $2^k$-Hadamard transform, $2^k$-crosscorrelation, $2^k$-autocorrelation and the related results for any $m\in\mathbb{N}$, establishing the existing results as specific cases within this generalized framework. We begin with extending the definition of $2^k$-Hadamard transform~\cite{2kbent} to define the most generalized version of the Hadamard transform, and connect it with the newly defined unitary operator $\Omega_m$.

\begin{definition}[$m$-Hadamard transform]
	\label{def:m-Hadamard}
	Given $f\in\mathcal{B}_n$, the $m$-Hadamard transform ($m\in\mathbb{N}$) of $f$ at $\bm{\omega}\in\mathbb{F}_2^n$ is a complex valued function, $\mathcal{H}^{(m)}_f: \mathbb{F}_2^n \rightarrow \mathbb{C}$ defined as
	$$\mathcal{H}^{(m)}_f(\bm{\omega})=2^{-n/2}\sum_{\var{x}\in\mathbb{F}_2^n}(-1)^{f(\var{x})\oplus \var{x}\cdot \bm{\omega}} \zeta_m^{wt(\var{x})}.$$
\end{definition}
\begin{remark}
	Note that for $m = 1$, the $m$-Hadamard transform reduces to the standard Walsh-Hadamard transform, i.e., $\mathcal{H}^{(1)}_f=W_f$. Similarly, for $m = 4$, it corresponds to the nega-Hadamard transform, $\mathcal{H}^{(4)}_f=N_f$. More generally, for $m=2^k$, the transform is expressed as $\mathcal{H}^{(m)}_f=\mathcal{H}_f^{(2^k)}$, representing the $2^k$-Hadamard transform.
\end{remark}

Since the $m$-Hadamard transform is a complex valued function, we can define the conjugate $m$-Hadamard transform as follows:
$$\overline{\mathcal{H}^{(m)}_f(\bm{\omega})} = 2^{-n/2}\sum_{\var{x}\in\mathbb{F}_2^n}(-1)^{f(\var{x})\oplus \var{x}\cdot \bm{\omega}} \left(\overline{\zeta_m}\right)^{wt(\var{x})}.$$
In this regard, the constraint $\sum_{\bm{\omega}\in\mathbb{F}_2^n}|\mathcal{H}^{(m)}_f(\bm{\omega})|^2 = \sum_{\bm{\omega}\in \mathbb{F}_2^n}\mathcal{H}^{(m)}_f(\bm{\omega})\overline{\mathcal{H}^{(m)}_f(\bm{\omega})}=2^n$ is known as the $m$-Parseval identity. The proof follows from Corollary~\ref{cor:m-cross-hadamard}.

Similar to the bent-negabent classifications, a Boolean function $f\in\mathcal{B}_n$ is called $m$-\textit{bent} if its $m$-Hadamard transform is flat in complex modulus, i.e., $\left|\mathcal{H}^{(m)}_f(\bm{\omega})\right| = 1$, for all $\bm{\omega} \in \mathbb{F}_2^n$. Furthermore, any Boolean function $f\in\mathcal{B}_n$ can be represented as a sum of its $m$-Hadamard transforms, as follows.
\begin{lemma}
\label{lem:m-spectra}
	Let $f\in\mathcal{B}_n$. Then the $m$-Hadamard transform is invertible, i.e.,
	$$(-1)^{f(\var{y})} = 2^{-n/2}\zeta_m^{-wt(\var{y})} \sum_{\bm{\omega} \in\mathbb{F}_2^n} \mathcal{H}^{(m)}_f(\bm{\omega})(-1)^{ \bm{\omega} \cdot \var{y}}.$$
\end{lemma}
\begin{proof} Let us begin with the right hand side (RHS)
	\begin{align*}
		2^{-n/2} \sum_{\bm{\omega} \in \mathbb{F}_2^n} \mathcal{H}^{(m)}_f(\bm{\omega})(-1)^{ \bm{\omega} \cdot \var{y}} =& 2^{-n} \sum_{\bm{\omega} \in \mathbb{F}_2^n} \sum_{\var{x}\in\mathbb{F}_2^n} (-1)^{f(\var{x})\oplus \var{x}\cdot \bm{\omega}}\zeta_m^{wt(\var{x})} (-1)^{ \bm{\omega} \cdot \var{y}}\\
		=& 2^{-n} \sum_{\var{x}\in\mathbb{F}_2^n} (-1)^{f(\var{x})}\zeta_m^{wt(\var{x})} \sum_{\bm{\omega}\in\mathbb{F}_2^n}  (-1)^{ \bm{\omega} \cdot (\var{x}\oplus \var{y})}\\
		=& 2^{-n} \sum_{\var{x}\in\mathbb{F}_2^n} (-1)^{f(\var{x})}\zeta_m^{wt(\var{x})} 2^n\delta_0(\var{x}\oplus \var{y})
		= (-1)^{f(\var{y})}\zeta_m^{wt(\var{y})},
	\end{align*}
and the claim is shown.
\end{proof}

Next, we compute the $m$-Hadamard transform for various combinations of Boolean functions (as in~\cite{pre,2kbent}).
\begin{lemma}
\label{lem:m-forr-prop}
Let $f,g,h\in\mathcal{B}_n$ and $\zeta_m = e^{2\pi i/m}$. Then, the following holds. 
    \begin{enumerate}[label=(\alph*)]
        \item If $g(\var{x}) = f(\xv) \op \var{c}\cdot \var{x} \oplus d$ with $\var{c}\in\mathbb{F}_2^n$, and $d\in\mathbb{F}_2$, then $\mathcal{H}_{g}^{(m)}({\var{u}}) = (-1)^d\mathcal{H}_{f}^{(m)}({\var{u}\oplus \var{c}})$. Moreover, if $\mathbb{L}(\xv) = \var{c}\cdot \xv \op d$, then $\mathcal{H}_{\mathbb{L}}^{(m)}({\var{u}}) = (-1)^d 2^{n/2} \left(\cos (\pi/m) \right)^n \left(-i\tan (\pi/m) \right)^{wt(\var{u} \oplus \var{c})} \zeta^{n/2 - wt(\var{u}\oplus \var{c})}$.
        \item If $f(\xv)=g(\xv)\oplus h(\xv)$, then $\mathcal{H}_{f}^{(m)}({\var{u}}) = 2^{-n/2}\sum_{\var{v}}\mathcal{H}^{(m)}_g(\var{v})W_h(\var{u}\oplus \var{v}) = 2^{-n/2}\sum_{\var{v}}W_g(\var{v})\mathcal{H}^{(m)}_h(\var{u}\oplus \var{v})$.
        \item If $g(\xv) = f(A\xv)$, where $A$ is an $n\times n$ orthogonal matrix over $\mathbb{F}_2$, then $\mathcal{H}_{g}^{(m)}(\var{u}) = \mathcal{H}_{f}^{(m)}(A\var{u})$.
        \item If $h(\xv,\yv)=f(\xv)\oplus g(\yv)$ with $\xv,\yv \in \mathbb{F}_2^n$, then $\mathcal{H}_{h}^{(m)}(\var{u},\var{v}) = \mathcal{H}_{f}^{(m)}(\var{u})\mathcal{H}_{g}^{(m)}(\var{v})$.
    \end{enumerate}
\end{lemma}
\begin{proof}
The proof follows the same lines as in~\cite[Theorem 2]{2kbent}, with some tricks and tweaks.
\begin{enumerate}[label=(\alph*)]
    \item From definition, $\mathcal{H}_{g}^{(m)}({\var{u}}) = 2^{-n/2}\sum_{\var{x}\in\mathbb{F}_2^n} (-1)^{f(\var{x}) \oplus \left(\var{u} \oplus \var{c}\right)\cdot \var{x} \oplus d} \zeta_m^{wt(\xv)} = (-1)^{d}\mathcal{H}_{f}^{(m)}({\var{u}\oplus \var{c}})$. Moreover,
\begin{align*}
1+\zeta_m &= 1 + \cos(2\pi/m) + i\sin (2\pi/m) 
= 2\cos^2(\pi/m) + 2i\cos(\pi/m)sin(\pi/m) 
= 2\cos(\pi/m)e^{i\pi/m},\\
1-\zeta_m &= 1 - \cos(2\pi/m) - i\sin (2\pi/m) 
= 2\sin^2(\pi/m) - 2i\cos(\pi/m)sin(\pi/m) 
= -2i\sin(\pi/m)e^{-i\pi/m}.
    \end{align*}
For $\var{u} = u_1,\ldots u_n$ and $\var{c} = c_1,\ldots c_n$, set $b_i = u_i\oplus c_i$. Then,
\begin{align*}
\mathcal{H}_{\mathbb{L}}^{(m)}({\var{u}}) = & 2^{-n/2}(-1)^d \sum_{\xv \in \mathbb{F}_2^n}(-1)^{\xv \cdot (\var{u} \oplus \var{c})}\zeta_m^{wt(\xv)}\\
= & 2^{-n/2} (-1)^d\prod_{i=1}^n\left( 1+ \zeta_m(-1)^{b_k}\right) = 2^{-n/2} (-1)^d\prod_{b_k=0}(1+\zeta_m)\prod_{b_k=1}(1 - \zeta_m)\\
= & 2^{-n/2}(-1)^d\left( 2\cos(\pi/m)e^{i\pi/m}\right)^{n-wt(\var{u}\oplus \var{c})} \left( -2i\sin(\pi/m)e^{-i\pi/m}\right)^{wt(\var{u}\oplus \var{c})}\\
= & (-1)^d 2^{n/2} \left(\cos(\pi/m) \right)^n \left(-i\tan(\pi/m) \right)^{wt(\var{u}\oplus \var{c})} \left(e^{2\pi i/m}\right)^{n/2-wt(\var{u}\oplus \var{c})}.
\end{align*}
    \item We prove the first identity by starting with the RHS. The second identity follows in a similar manner.
    \begin{align*}
        2^{-n/2}\sum_{\var{v}\in\mathbb{F}_2^n}\mathcal{H}^{(m)}_g(\var{v})W_h(\var{u}\oplus \var{v}) = & 2^{-3n/2}\sum_{\var{v},\var{x},\var{y}\in\mathbb{F}_2^n}(-1)^{g(\xv)\oplus h(\var{y}) \oplus \var{x}\cdot \var{v}\oplus \var{y}\cdot(\var{u}\oplus \var{v})}\zeta_m^{wt(\xv)}\\
        = & 2^{-3n/2}\sum_{\var{x},\var{y}\in\mathbb{F}_2^n}(-1)^{g(\xv)\oplus h(\var{y})\oplus \var{u}\cdot \var{y}} \zeta_m^{wt(\xv)}\sum_{\var{v}\in\mathbb{F}_2^n}(-1)^{\var{v}\cdot (\var{x}\oplus \var{y})}\\
        = & 2^{-n/2}\sum_{\var{x}\in\mathbb{F}_2^n}(-1)^{g(\xv)\oplus h(\var{x})\oplus \var{u}\cdot \var{x}} \zeta_m^{wt(\xv)} = \mathcal{H}_{g\oplus h}^{(m)}({\var{u}}) = \mathcal{H}_{f}^{(m)}({\var{u}}) .
    \end{align*}
    \item From definition, $\mathcal{H}_{g}^{(m)}(\var{u}) = 2^{-n/2}\sum_{\var{x}\in\mathbb{F}_2^n}(-1)^{f(A\xv) \oplus \xv\cdot \var{u}}\zeta_m^{wt(\xv)}$. Suppose, $A\xv = \var{y}$, then $\xv = A^T\var{y}$. Thus,
    $$\mathcal{H}_{g}^{(m)}(\var{u}) = 2^{-n/2}\sum_{\var{y}\in\mathbb{F}_2^n}(-1)^{f(\var{y}) \oplus (A^T\var{y})\cdot \var{u}}\zeta_m^{wt(A^T\var{y})}= 2^{-n/2}\sum_{\var{y}\in\mathbb{F}_2^n}(-1)^{f(\var{y}) \oplus \var{y}\cdot A\var{u}}\zeta_m^{wt(\var{y})} = \mathcal{H}_{f}^{(m)}(A\var{u}).$$
    \item Starting with the RHS, we obtain
    \begin{align*}
        \mathcal{H}_{f}^{(m)}(\var{u})\mathcal{H}_{g}^{(m)}(\var{v}) = & 2^{-n}\sum_{\var{x},\var{y}\in\mathbb{F}_2^n}(-1)^{f(\xv)\oplus g(\yv) \oplus \var{x}\cdot \var{u} \oplus \var{y}\cdot \var{v}}\zeta_m^{wt(\xv)+wt(\yv)}\\
        = & 2^{-n}\sum_{\var{x},\var{y}\in\mathbb{F}_2^n}(-1)^{h(\xv,\yv) \oplus (\var{x},\var{y})\cdot(\var{u}, \var{v})}\zeta_m^{wt(\xv,\yv)} = \mathcal{H}_{h}^{(m)}(\var{u},\var{v}).
    \end{align*}
\end{enumerate}
The claims are shown.
\end{proof}

Next, we define the most generalized version of \textit{crosscorrelation} and \textit{autocorrelation}, as follows.

\begin{definition}[$m$-\textit{crosscorrelation}]
	\label{def:m-crosscorrelation}
	Given $f,g\in\mathcal{B}_n$, the $m$-\textit{crosscorrelation} of $f$ and $g$ at a point $\var{y}\in\mathbb{F}_2^n$ is described as
	$$C_{f,g}^{(m)}(\var{y}) = \sum_{\xv\in\mathbb{F}_2^n}(-1)^{f(\xv)\oplus g(\xv\oplus \var{y})}\left(\zeta_{m}^2\right)^{\xv\odot\var{y}}.$$
	Taking $f=g$ in the above formulation, we obtain the $m$-\textit{autocorrelation}, defined as 
    $$C_{f}^{(m)}(\var{y}) = \sum_{\xv\in\mathbb{F}_2^n}(-1)^{f(\xv)\oplus f(\xv\oplus \var{y})}\left(\zeta_{m}^2\right)^{\xv\odot\var{y}},$$
	where $\xv\odot\var{y}$ denotes the inner product in $\mathbb{C}\times \mathbb{C}$. 
\end{definition}
It is straightforward to observe that for $m=1$, the $m$-crosscorrelation ($m$-autocorrelation) reduces to the standard crosscorrelation (autocorrelation). Additionally, for $m=4$, where $\zeta_m=i$, the $m$-crosscorrelation and $m$-autocorrelation correspond to nega-crosscorrelation and nega-autocorrelation, respectively. The $m$-crosscorrelation of two Boolean functions is related with the product of their $m$-Hadamard transforms, as follows.
\begin{theorem}
	\label{thm:m-cross-hadamard}
	Let $f,g\in \mathcal{B}_n$. Then
	$$C_{f,g}^{(m)}(\var{z})=\zeta_{m}^{wt(\var{z})}\sum_{\var{u}\in\mathbb{F}_2^n}\mathcal{H}_{f}^{(m)}(\var{u})\overline{\mathcal{H}_{g}^{(m)}(\var{u})}(-1)^{\var{u}\cdot\var{z}}.$$
\end{theorem}
\begin{proof}
Let us begin with the RHS, namely,
\begin{align*}
&\zeta_{m}^{wt(\var{z})}\sum_{\var{u}\in\mathbb{F}_2^n}\mathcal{H}_{f}^{(m)}(\var{u})\overline{\mathcal{H}_{g}^{(m)}(\var{u})}(-1)^{\var{u}\cdot\var{z}}\\
= &2^{-n}\zeta_{m}^{wt(\var{z})}\sum_{\var{u}\in\mathbb{F}_2^n} \sum_{\var{x}\in\mathbb{F}_2^n}(-1)^{f(\var{x})\oplus \var{x}\cdot \var{u}} \zeta_m^{wt(\var{x})} \sum_{\var{y}\in\mathbb{F}_2^n}(-1)^{g(\var{y})\oplus \var{y}\cdot \var{u}} \zeta_m^{-wt(\var{y})}(-1)^{\var{u}\cdot\var{z}}\\
= &2^{-n}\sum_{\var{x},\var{y}\in\mathbb{F}_2^n}(-1)^{f(\var{x})\oplus g(\var{y})} \zeta_m^{wt(\var{x})+ wt(\var{z}) - wt(\var{y})}\sum_{\var{u}\in\mathbb{F}_2^n} (-1)^{\var{x}\cdot\var{u} \oplus \var{y}\cdot\var{u} \oplus \var{u}\cdot\var{z}}.
\end{align*}
Now taking $\var{y}=\var{x}\oplus \var{z}$ in the above expression:
\begin{align*}
 &2^{-n}\sum_{\var{x}\in\mathbb{F}_2^n}(-1)^{f(\var{x})\oplus g(\var{x}\oplus \var{z})} \zeta_m^{wt(\var{x})+ wt(\var{z}) - wt(\var{x}\oplus \var{z})} \sum_{\var{u}\in\mathbb{F}_2^n}(-1)^{\var{u}\cdot \var{y} \oplus \var{u}\cdot\left(\var{x}\oplus \var{z}\right)}\\
= &\sum_{\var{x}\in\mathbb{F}_2^n}(-1)^{f(\var{x})\oplus g(\var{x}\oplus \var{z})} \left(\zeta^2_m\right)^{\var{x}\odot \var{z}} = C_{f,g}^{(m)}(\var{z}).
	\end{align*}
\end{proof}

\begin{corollary}
	\label{cor:m-cross-hadamard}
	Let $f\in\mathcal{B}_n$. Then  $C_{f}^{(m)}(\var{z})=\zeta_{m}^{wt(\var{z})}\sum_{\var{u}\in\mathbb{F}_2^n}|\mathcal{H}_{f}^{(m)}(\var{u})|^2(-1)^{\var{u}\cdot\var{z}}.$
\end{corollary}
Assuming $\var{z}=0^n$ in the above corollary, we obtain $\sum_{\var{u}\in\mathbb{F}_2^n}|\mathcal{H}_{f}^{(m)}(\var{u})|^2 = C_{f}^{(m)}(0^n) = 2^n$, establishing the $m$-Parseval's identity. Moreover, the $m$-bent criterion of a Boolean function is related to its $m$-autocorrelation as follows.
\begin{theorem}
	\label{thm:autobent}
	A Boolean function $f\in\mathcal{B}_n$ is $m$-bent if and only if $C_{f}^{(m)}(\var{z})= 0 $ for all $\var{z}\in\mathbb{F}_2^n\setminus\{0^n\}$, and $C_{f}^{(m)}(\var{z})=  2^n$, for $\var{z}=0^n$.
\end{theorem}
\begin{proof}
	If $f\in\mathcal{B}_n$ is $m$-bent, i.e., $|\mathcal{H}_{f}^{(m)}(\var{u})| = 1$ for all $\var{u}\in\mathbb{F}_2^n$, then $C_{f}^{(m)}(\var{z})=\zeta_{m}^{wt(\var{z})}\sum_{\var{u}\in\mathbb{F}_2^n}(-1)^{\var{u}\cdot\var{z}}$. Now, $\sum_{\var{u}\in\mathbb{F}_2^n}(-1)^{\var{u}\cdot\var{z}} = 0$ for all $\var{z}\in\mathbb{F}_2^n\setminus\{0^n\}$, and $2^n$, if $\var{z}=0^n$. Similarly, the converse follows.
\end{proof}

Next, we present the most generalized version of a Deutsch-Jozsa-like quantum algorithm, building upon the existing variations as shown in~\cite{dj,exdj}.

\subsection{Generalized Deutsch-Jozsa algorithm} 
Suppose, we are given the oracle access to an unknown Boolean function $f\in \mathcal{B}_n$. Similar to the standard Deutsch-Jozsa algorithm, we start with an $(n+1)$-qubit quantum circuit initialized to $\ket{0^n}\ket{-}$. We then apply $n$-many Hadamard gates to the first $n$ qubits, creating an equal superposition of all possible $n$-bit strings, $2^{-n/2}\sum_{\xv\in\mathbb{F}_2^n}\ket{\var{x}}$, which is fed into the quantum oracle $U_f$, thereby achieving quantum parallelism.

After the oracle query, instead of only applying the Hadamard gates to the first $n$ qubits, as in the standard Deutsch-Jozsa algorithm, we apply a sequence of omega-gates determined by a sequence of natural numbers $\var{d}=(d_1,\ldots ,d_n)$ where $d_i\in\mathbb{N}$ dictates the application of $\Omega_{d_i}$ gate at the $i$-th qubit. Finally, we measure the first $n$ qubits in the $\{\ket{0},\ket{1}\}$ basis, and the final state before measurement is given by
$$2^{-n}\sum_{\substack{\xv\in\mathbb{F}_2^n, \,d_i\in\mathbb{N} }}\sum_{\var{y}\in\mathbb{F}_2^n}(-1)^{f(\xv)\oplus \xv\cdot \var{y}}\zeta_{d_i}^{wt(\xv)}\ket{\var{y}}.$$
Therefore, the probability of observing any particular state $\var{y}\in\mathbb{F}_2^n$ becomes 
$$\pr{(\var{y})} = 2^{-2n}\left|\sum_{\substack{\xv\in\mathbb{F}_2^n, \, d_i\in\mathbb{N}}}(-1)^{f(\xv)\oplus \xv\cdot \var{y}}\zeta_{d_i}^{wt(\xv)}\right|^2.$$
A schematic diagram of the generalized Deutsch-Jozsa algorithm is presented in Fig.~\textup{\ref{fig:ndj}}, where $U_i=\Omega_{d_i}$ and $d_i\in\mathbb{N}$. If $d_i=m$ for all $1\leq i \leq n$ where $m$ is a fixed natural number, then we denote the generalized Deutsch-Jozsa algorithm by $\text{DJ}_{m}$.
\begin{figure}[ht]
	\centering
	\includegraphics[scale=1.1]{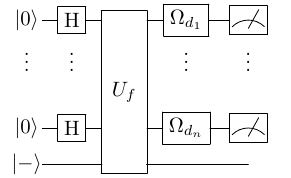}
	\caption{Quantum circuit for the most generalized Deutsch-Jozsa algorithm.}
	\label{fig:ndj}
\end{figure}

\begin{remark}
Depending on the exact sequence of natural numbers, $\var{d}$, we have the following observations.
\begin{enumerate}
    \item If $d_i=1$ for all $1\leq i\leq n$, i.e., $U_i=\Omega_1=\mathrm{H}$ for all $1\leq i\leq n$, then the algorithm $\text{DJ}_{1}$ corresponds to the standard Deutsch-Jozsa algorithm~\textup{\cite{dj}}, where the final state before measurement becomes $2^{-n}\sum_{\xv,\var{y}\in\mathbb{F}_2^n}(-1)^{f(\xv)\oplus \xv\cdot \var{y}}\ket{\var{y}} = 2^{-n/2}\sum_{\var{y}\in\mathbb{F}_2^n}W_f(\var{y})\ket{\var{y}}$, which is the normalized Walsh-Hadamard transform of $f$. As a result, the probability of observing any particular state $\var{y}\in\mathbb{F}_2^n$ becomes $2^{-n}\left|W_f(\var{y}) \right|^2$.
    \item When $d_i\in\{1,4\}$, i.e., $U_i\in\{\mathrm{H},\mathrm{N}\}$ for all $1\leq i\leq n$, then depending upon the choice of $\var{d}\in \{1,4\}^n$, a combination of the Hadamard and the nega-Hadamard gates are applied, resulting in the extended Deutsch-Jozsa algorithm as introduced in~\textup{\cite{exdj}}. Recall that in the original extended Deutsch-Jozsa algorithm, the choice between Hadamard and nega-Hadamard gates was determined by a binary sequence $\var{c}=(c_j)_{1\leq j\leq n} \in\mathbb{F}_2^n$, where $c_j=0$ corresponds to the Hadamard gate and $c_j=1$ corresponds to the nega-Hadamard gate. Similarly, in this context, $d_j=1$ implies $c_j=0$, while $d_j=4$ corresponds to $c_i=1$. Additionally, fixing $d_j=4$ for all $1\leq j\leq n$ (algorithm $\mathrm{DJ}_{4}$) the final state before measurement becomes $2^{-n}\sum_{\xv,\var{y}\in\mathbb{F}_2^n}(-1)^{f(\xv)\oplus \xv\cdot \var{y}}i^{wt(\xv)}\ket{\var{y}} = 2^{-n/2}\sum_{\var{y}\in\mathbb{F}_2^n}N_f(\var{y})\ket{\var{y}}$, which is the normalized nega-Hadamard transform of $f$. Consequently, the probability of observing any particular state $\var{y}\in\mathbb{F}_2^n$ is given by $2^{-n}\left|N_f(\var{y}) \right|^2$.
    \item Similarly, if $d_j=2^k$ for all $1\leq j\leq n$ and some fixed $k\in\mathbb{N}$ (corresponds to $\mathrm{DJ}_{2^k}$), then the final state before measurement becomes the normalized $2^k$-Hadamard transform, and the corresponding probability of observing any particular state $\var{y}\in\mathbb{F}_2^n$ becomes $2^{-n}\left|\mathcal{H}_f^{(2^k)}\right|^2$.
    \item Finally, in algorithm $\mathrm{DJ}_{m}$, where $d_j=m$ for all $1\leq j \leq n$ and some fixed $m\in\mathbb{N}$, the final pre-measurement state becomes the normalized $m$-Hadamard transform, and the probability of observing any particular state $\var{y}\in\mathbb{F}_2^n$ is given by $2^{-n}\left|\mathcal{H}_f^{(m)}\right|^2$.
	\end{enumerate}
\end{remark}
\subsection{$m$-Forrelation}
We now present the most generalized form of Forrelation, called $m$-Forrelation, which encompasses both the standard ($3$-fold) Forrelation~\cite{forr} and the nega-Forrelation~\cite{dcc} as specific instances.
\begin{definition}[$m$-Forrelation]
	\label{def:m-forr}
Given oracle access to $f_1,f_2,f_3\in\mathcal{B}_n$, the ($3$-fold) $m$-Forrelation measures a correlation between the Boolean function $f_1$, the $m$-Hadamard transform of $f_2$ and the conjugate $m$-Hadamard transform of $f_3$, precisely defined as 
$$\Phi^{(m)}_{f_1,f_2,f_3}=\displaystyle\frac{1}{2^{n}}\sum_{\var{x_1}\in\mathbb{F}_2^n}{(-1)^{f_1(\var{x_1})}\mathcal{H}^{(m)}_{f_2}(\var{x_1})\overline{\mathcal{H}^{(m)}_{f_3}(\var{x_1})}},$$
	which can be further decomposed to
	\begin{align*}
		\displaystyle\frac{1}{2^{2n}}\sum\limits_{\var{x_1}, \var{x_2} , \var{x_3}\in \mathbb{F}_2^n}
		(-1)^{f_1(\var{x_1})} \left((-1)^{f_2(\var{x_2})\oplus \var{x_1}\cdot \var{x_2}} \zeta_m^{wt(\var{x_2})}\right) \left((-1)^{f_3(\var{x_3})\oplus \var{x_1}\cdot \var{x_3}}(\overline{\zeta_m})^{wt(\var{x_3})}\right).
	\end{align*}
\end{definition}

\begin{remark}
	From the definition of ($3$-fold) $m$-Forrelation, we have the following remarks:
	\begin{enumerate}
		\item The ($3$-fold) $m$-Forrelation, $\Phi^{(m)}_{f_1,f_2,f_3}$ is complex-valued and not symmetric, meaning its values depend on the order of the Boolean functions.
		\item When $f_2=f_3$, the product of the $m$-Hadamard transform and its conjugate equals the complex square of the $m$-Hadamard transform of $f$. Consequently, $\Phi^{(m)}_{f_1,f_2,f_2}$ is always a real number.
		\item For $m=1$, $\Phi^{(m)}_{f_1,f_2,f_3} = \Phi_{f_1,f_2,f_3}$, the standard ($3$-fold) Forrelation as provided in~\cite{forr}. For $m=4$, $\Phi^{(m)}_{f_1,f_2,f_3} = \eta_{f_1,f_2,f_3}$, the $3$-fold nega-Forrelation as introduced in~\textup{\cite{dcc}}.
	\end{enumerate}
\end{remark}

Additionally, note that, similar to the Forrelation, the $m$-Forrelation formulation can be extended to accommodate $k$ many Boolean functions, $f_1,\ldots,f_k\in\mathcal{B}_n$, referred to as $k$-fold $m$-Forrelation. However, in this paper, we focus mainly on the $3$-fold variation and simply use the term $m$-Forrelation to refer to the $3$-fold $m$-Forrelation. Following the approach in~\cite{forr,dcc}, we present two quantum algorithms for estimating the $m$-Forrelation values—one utilizing three sequential queries and another using two parallel queries.

We begin with the $3$-query quantum algorithm (see Fig.~\ref{fig:mforr-q3f3}). Given oracle access to the Boolean functions $f_1, f_2,f_3\in\mathcal{B}_n$, we begin with the state $\ket{0^n}\ket{-}$ and traverse through the following sequence of steps,
$$\mathrm{H}^{\otimes n}\rightarrow  U_{f_2}\rightarrow \Omega_m^{\otimes n}\rightarrow U_{f_1}\rightarrow \mathrm{H}^{\otimes n}\rightarrow  U_{f_3} \rightarrow \overline{\Omega_m}^{\otimes n}.$$
%where all the $n$-qubit gates $\left(H^{\otimes n},\Omega_m^{\otimes n},\overline{\Omega}_m^{\otimes n}\right)$ are applied to the $n$ query-qubits ($q_1,\ldots ,q_n$ in Fig.~\ref{fig:q3f3T}) and the oracles $\left(U_{f_2},U_{f_1},U_{f_3} \right)$ are applied to all the $n+1$ qubits ($q_1,\ldots,q_n,q_{n+1}$ in Fig.~\ref{fig:q3f3T}).
Ignoring the last qubit, the pre-measurement amplitude corresponding to the state $\ket{0^n}$ becomes
$$\frac{1}{2^{2n}}\displaystyle\sum_{\var{x_1},\var{x_2},\var{x_3}\in\mathbb{F}_2^n}(-1)^{f_2(\var{x_2})}\zeta_m^{wt(\var{x_2})}(-1)^{\var{x_1}\cdot \var{x_2}}(-1)^{f_1(\var{x_1})}(-1)^{\var{x_1}\cdot \var{x_3}}(\overline{\zeta_m})^{wt(\var{x_3})}(-1)^{f_3(\var{x_3})},$$
which is equal to $\Phi^{(m)}_{f_1,f_2,f_3}$.
Since, $\Phi^{(m)}_{f_1,f_2,f_3}$ is a complex number, the probability of observing the all-zero state upon measurement is given by the complex modulus square, $\left|\Phi^{(m)}_{f_1,f_2,f_3}\right|^2$. Let us denote the $3$-query $m$-Forrelation algorithm by $A^{(m)3,3}_n$. Fig.~\ref{fig:mforr-q3f3} provides a schematic diagram of the quantum circuit for $A_n^{(m)3,3}(f_1,f_2,f_3)$.
\begin{figure}[ht]
	\centering
	\includegraphics[scale=1.2]{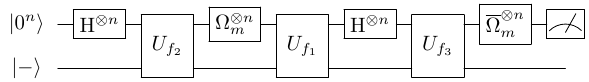}
	\caption{The $3$-query quantum circuit for estimating $m$-Forrelation, $\Phi^{(m)}_{f_1,f_2,f_3}$.}
	\label{fig:mforr-q3f3}
\end{figure}

Next, we present the $2$-query quantum algorithm for estimating the $m$-Forrelation (Fig.~\ref{fig:mforr-q2f3}). Given oracle access to $f_1, f_2,f_3\in\mathcal{B}_n$, we begin with an $(n+2)$-qubit state, $\ket{+}\ket{0^n}\ket{-}$, where the first qubit is termed as the `driving qubit' and the next $n$ many as the query-qubits. In the beginning, $n$ many Hadamard gates are applied to the $n$-query qubits, and distribute the state as follows:
$$\textstyle\ket{+}\ket{0^n}\ket{-} \xrightarrow{\mathrm{H}^{\otimes n}} \frac{1}{\sqrt{2^{n+1}}} \left(\sum_{\var{x_2}\in \mathbb{F}_2^n}\ket{0}\ket{\var{x_2}} + \sum_{\var{x_3}\in\mathbb{F}_2^n}\ket{1}\ket{\var{x_3}}\right) \ket{-}.$$
Then controlled on the driving qubit being $\ket{0}$, we sequentially apply $U_{f_2}\rightarrow \Omega_m^{\otimes n}\rightarrow U_{f_1} \rightarrow \mathrm{H}^{\otimes n}$ and obtain
$$\frac{\ket{0}}{\sqrt{2^{3n+1}}} \displaystyle \sum_{\var{x_1},\var{x_2},\var{x_3} \in\mathbb{F}_2^n} (-1)^{f_2(\var{x_2})}\zeta_m^{wt\left(\var{x_2}\right)}(-1)^{\var{x_1}\cdot \var{x_2}}(-1)^{f_1(\var{x_1})}(-1)^{\var{x_1}\cdot \var{x_3}}\ket{\var{x_3}}\ket{-}.$$
Similarly, controlled on the driving qubit being $\ket{1}$, we sequentially apply: $\mathrm{S}_m^{\otimes n}\rightarrow U_{f_3}$ and obtain
$$\frac{\ket{1}}{\sqrt{2^{n+1}}} \displaystyle \sum_{\var{x_3}\in \mathbb{F}_2^n}(-1)^{f_3(\var{x_3})}\zeta_m^{wt(\var{x_3})}\ket{\var{x_3}}\ket{-}.$$
Combining both the scenario, we obtain $\textstyle\ket{\psi}=\sum_{\var{x_3}\in\mathbb{F}_2^n}\left(\alpha_{\var{x_3}}\ket{0} + \beta_{\var{x_3}}\ket{1}\right)\ket{\var{x_3}} \ket{-}$, where 
\begin{align*} 
&\alpha_{\var{x_3}}=\left(\frac{1}{\sqrt{2^{3n+1}}} \sum_{\var{x_1},\var{x_2}\in \mathbb{F}_2^n} (-1)^{f_2(\var{x_2})} \zeta_m^{wt\left(\var{x_2}\right)}(-1)^{\var{x_1}\cdot \var{x_2}} (-1)^{f_1(\var{x_1})}(-1)^{\var{x_1}\cdot \var{x_3}} \right) \text{ and } \\
& \displaystyle\beta_{\var{x_3}} = \frac{1}{\sqrt{2^{n+1}}} (-1)^{f_3(\var{x_3})} \zeta_m^{wt\left(\var{x_3}\right)}.
\end{align*}
Next, we apply a Hadamard gate to the `driving qubit', and obtain
$$\frac{1}{\sqrt{2}} \left(\displaystyle \sum_{\var{x_3}\in \mathbb{F}_2^n}\left(\alpha_{\var{x_3}} + \beta_{\var{x_3}}\right)\ket{0} + \displaystyle\sum_{\var{x_3}\in \mathbb{F}_2^n} \left(\alpha_{\var{x_3}} - \beta_{\var{x_3}}\right)\ket{1}\right)\ket{\var{x_3}}\ket{-}.$$
Finally, we measure the `driving qubit' in the computational basis. The probability of observing $\ket{0}$ is given by
$$
\frac{1}{2}\displaystyle\sum_{\var{x_3}\in\mathbb{F}_2^n} \left|\alpha_{\var{x_3}} + \beta_{\var{x_3}} \right|^2 = \frac{1}{2}\displaystyle\left[\sum_{\var{x_3}\in\mathbb{F}_2^n} \left(\left|\alpha_{\var{x_3}} \right|^2 + \left|\beta_{\var{x_3}} \right|^2\right) + 2 \Re{\left(\alpha_{\var{x_3}}\overline{\beta}_{\var{x_3}}\right)}\right],
$$
where $\Re{\left(z\right)}$ denotes the real part of the complex number $z$. Note that the expression
$\sum_{\var{x_3}\in\mathbb{F}_2^n}\left|\alpha_{\var{x_3}}\right|^2 + \left|\beta_{\var{x_3}}\right|^2$ represents the sum of the squared amplitudes of the quantum state $\ket{\psi}$, which is equal to $1$. Moreover,
\allowdisplaybreaks
\begin{align*}
\sum_{\var{x_3}\in\mathbb{F}_2^n} 2\Re{\left(\alpha_{\var{x_3}}\overline{\beta}_{\var{x_3}}\right)} = & 2\Re{\left(\sum_{\var{x_3}\in\mathbb{F}_2^n} \alpha_{\var{x_3}}\overline{\beta}_{\var{x_3}}\right)}\\
= & \Re{\left( \frac{1}{2^{2n}}\sum_{\var{x_1,x_2,x_3}\in\mathbb{F}_2^n} (-1)^{f_2(\var{x_2})} \zeta_m^{wt(\var{x_2})}(-1)^{\var{x_1}\cdot \var{x_2}} (-1)^{f_1(\var{x_1})} (-1)^{\var{x_1}\cdot \var{x_3}} (\overline{\zeta_m})^{wt(\var{x_3})}(-1)^{f_3(\var{x_3})} \right)}\\
= & \Re{\left( \Phi^{(m)}_{f_1,f_2,f_3} \right) }.
\end{align*}
Therefore, the probability of observing $\ket{0}$ upon measuring the driving qubit is given by $\frac{1}{2}\left(1+\Re{\left( \Phi^{(m)}_{f_1,f_2,f_3} \right)}\right)$. Let us denote the $2$-query $m$-Forrelation algorithm as $A_n^{(m)2,3}$. In $A_n^{(m)2,3}$, since $U_{f_3}$ is applied in parallel with $U_{f_2}$ and $U_{f_1}$, the effective query complexity remains $2$, and hence called the $2$-query algorithm.
\begin{remark}
Note that, in the 3-query $m$-fold $m$-Forrelation circuit (Fig.~\textup{\ref{fig:mforr-q3f3}}), the second and fourth layers employ the $\Omega_{m}$ and $\overline{\Omega}_m$ gates, respectively, in place of the nega-Hadamard and conjugate nega-Hadamard gates used in nega-Forrelation, or the Hadamard layers used in standard Forrelation algorithms.

Similarly, in the 2-query 3-fold $m$-Forrelation circuit (Fig.~\textup{\ref{fig:mforr-q2f3}}), a series of $\Omega_{m}$ gates is applied between the oracles for $f_1$ and $f_2$, replacing the nega-Hadamard gates found in nega-Forrelation. These modifications are essential to ensure that the measurement probabilities reflect the $m$-Forrelation spectra.
\end{remark}
Fig.~\ref{fig:mforr-q2f3} provides a schematic diagram of the quantum circuit for $A_n^{(m)2,3}(f_1,f_2,f_3)$.
\begin{figure}[ht]
	\centering
	\includegraphics[scale=1.2]{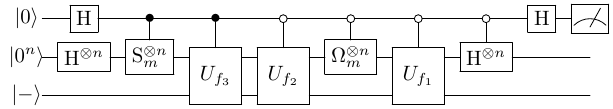}
	\caption{The $2$-query quantum circuit for estimating $m$-Forrelation, $\Phi^{(m)}_{f_1,f_2,f_3}$.}
	\label{fig:mforr-q2f3}
\end{figure}

\begin{remark}
	Since the $m$-Forrelation value, $\Phi^{(m)}_{f_1,f_2,f_3}$ is a complex number for $m>2$, using the $2$-query algorithm, we can estimate only the real part of $\Phi^{(m)}_{f_1,f_2,f_3}$ and not the complete $m$-Forrelation value. Furthermore, since $\Phi^{(m)}_{f_1,f_2,f_3}$ is real when $f_2=f_3$, the probability of observing $\ket{0}$ upon measuring the driving qubit after executing $A_n^{(m)2,3}(f_1,f_2,f_2)$ is given by $\frac{1}{2}\left(1 + \Phi^{(m)}_{f_1,f_2,f_2} \right)$. Consequently, the probability of observing $\ket{1}$ is $\frac{1}{2}\left(1 - \Phi^{(m)}_{f_1,f_2,f_2} \right)$.
\end{remark}

In the next section, we present different strategies for sampling the $m$-Hadamard transform and $m$-crosscorrelation (consequently $m$-autocorrelation) values using the $m$-Forrelation algorithms, $A_n^{(m)3,3}$ and $A_n^{(m)2,3}$.

\section{Sampling of generalized spectra using $m$-Forrelation}
\label{sec:cont2}
Given $f\in\mathcal{B}_n$, and a set of points $S\subseteq \mathbb{F}_2^n$, the goal is to estimate the $m$-Hadamard transform values of $f$ at all the points in $S$. Recall that using algorithm $\text{DJ}_{m}$, one can sample the $m$-Hadamard transform of $f$ at $S\subseteq \mathbb{F}_2^n$ with probability 
$\frac{1}{2^n}\sum_{\xv\in S}\left|\mathcal{H}^{(m)}_{f}(\var{x})\right|^2=: p$ (say),
where $0\leq p\leq 1$. We now present the strategies for sampling the $m$-Hadamard transforms of $f$ using the $m$-Forrelation algorithms, $A_n^{(m)3,3}$ and $A_n^{(m)2,3}$.

Suppose, $g\in\mathcal{B}_n$ such that $g(\var{x})=1$ for all $\var{x}\in S$,
and $g(\var{x})=0$ otherwise. From the definition of $\Phi^{(m)}_{f_1,f_2,f_3}$, with $f_2=f_3=f$ and $f_1=g$, we have
$$\Phi^{(m)}_{g,f,f}=\frac{1}{2^{n}}\sum_{\var{x}\in\mathbb{F}_2^n}{(-1)^{g(\var{x})}\left|\mathcal{H}^{(m)}_{f}(\var{x})\right|^2}=\frac{1}{2^{n}}\left(\sum_{\var{x}\not\in S}{\left|\mathcal{H}^{(m)}_{f}(\var{x})\right|^2}-\sum_{\var{x}\in S}{\left|\mathcal{H}^{(m)}_{f}(\var{x})\right|^2}\right).$$
Using the $m$-Parseval's identity, $\sum_{\var{x}\in\mathbb{F}_2^n}{\left|\mathcal{H}^{(m)}_{f}(\var{x})\right|^2}=2^n$ we obtain,
\begin{align*}
	\Phi^{(m)}_{g,f,f}=& \frac{1}{2^n}\left(2^n -\sum_{\var{x}\in S}{\left|\mathcal{H}^{(m)}_{f}(\var{x})\right|^2} - \sum_{\var{x}\in S}{\left|\mathcal{H}^{(m)}_{f}(\var{x})\right|^2}\right)
	= 1-\frac{2}{2^n}\left(\sum_{\var{x}\in S}{\left|\mathcal{H}^{(m)}_{f}(\var{x})\right|^2} \right)=1-2p.
\end{align*}
This implies, $p=\frac{1}{2}\left(1 - \Phi^{(m)}_{g,f,f} \right)$,
which is same as the probability of observing $\ket{1}$ upon measuring the `driving qubit' from running the algorithm $A_n^{(m)2,3}(g,f,f)$.
\begin{proposition}
	\label{prop:sample}
	Given $f,g\in\mathcal{B}_n$ and a set of points $S\subseteq \mathbb{F}_2^n$ such that $g(\xv)=1$ for all $\xv\in S$ and $g(\xv)=0$ otherwise, the probability of observing $\ket{1}$ upon measuring the driving qubit, from executing $A_n^{(m)2,3}(g,f,f)$ is given by $p$.
\end{proposition}

Therefore, the $2$-query $m$-Forrelation algorithm $A_n^{(m)2,3}(g,f,f)$ samples the $m$-Hadamard transform with a probability exactly equal to the sampling probability of Algorithm $\text{DJ}_{m}$.

Next, we compute the sampling probability of the $m$-Hadamard transform using the $3$-query algorithm, $A_n^{(m)3,3}$. Upon measurement, the probability of observing the all-zero state is given by $\left(\Phi^{(m)}_{g,f,f}\right)^2$. Consequently, the probability of observing a state with at least one $\ket{1}$ in the output is $1-\left(\Phi^{(m)}_{g,f,f}\right)^2$. Substituting $\Phi^{(m)}_{g,f,f}=1-2p$, we obtain:
$$1-\left(\Phi^{(m)}_{g,f,f}\right)^2=1-\left(1-2p \right)^2=4p-4p^2 \approx 4p.$$
Therefore, we have the following theorem.
\begin{theorem}
	\label{thm:sample-mHadamard}
	Given $f,g\in\mathcal{B}_n$ and a set of points $S\subseteq \mathbb{F}_2^n$ such that $g(\xv)=1$ for all $\xv\in S$ and $g(\xv)=0$ otherwise, the probability of one of the states in the measurement outcomes being $\ket{1}$ from executing the $3$-query quantum algorithm $A_n^{(m)3,3}(g,f,f)$ is given by $4p-4p^2$. 
\end{theorem}

From Theorem~\ref{thm:sample-mHadamard}, it is evident that when $p<0.75$ (i.e., when $p<4p-4p^2$), the sampling probability obtained from $A_n^{(m)3,3}$ surpasses that of algorithm $\text{DJ}_{m}$, in terms of the required number of queries. More specifically, for any $f\in\mathcal{B}_n$, the algorithm $A_n^{(m)3,3}$ samples the smaller values of the $m$-Hadamard transform more efficiently than the algorithm $\text{DJ}_{m}$. Conversely, for $p>0.75$, algorithm $\text{DJ}_{m}$ alone is sufficient for estimating the $m$-Hadamard transform values.

Here, one might argue that the algorithm $\text{DJ}_{m}$ requires only a single query to the oracle $U_f$, whereas the algorithm $A_n^{(m)3,3}$ necessitates two queries to $U_f$. However, even after executing $\text{DJ}_{m}$ twice (equivalent to two queries to $U_f$), the resulting sampling probability, $1-(1-p)^2 = 2p-p^2 \approx 2p$, remains lower than that achieved by $A_n^{(m)3,3}$. Furthermore, the sampling probability obtained from $A_n^{(m)3,3}$ also exceeds that of performing $\text{DJ}_{m}$ once, followed by a single round of amplitude amplification, which likewise requires two queries to $U_f$, and the corresponding sampling probability is given by $\sin\left(3\sin^{-1} p\right)\approx 3p$. For a graphical comparison of the sampling probabilities across different approaches, refer to Fig.~\ref{fig:graph}, which naturally resembles~\cite[Fig. 4]{amc} and~\cite[Fig. 3]{dcc}.
\begin{figure}[ht]
	\begin{center}
		\includegraphics[scale=0.55]{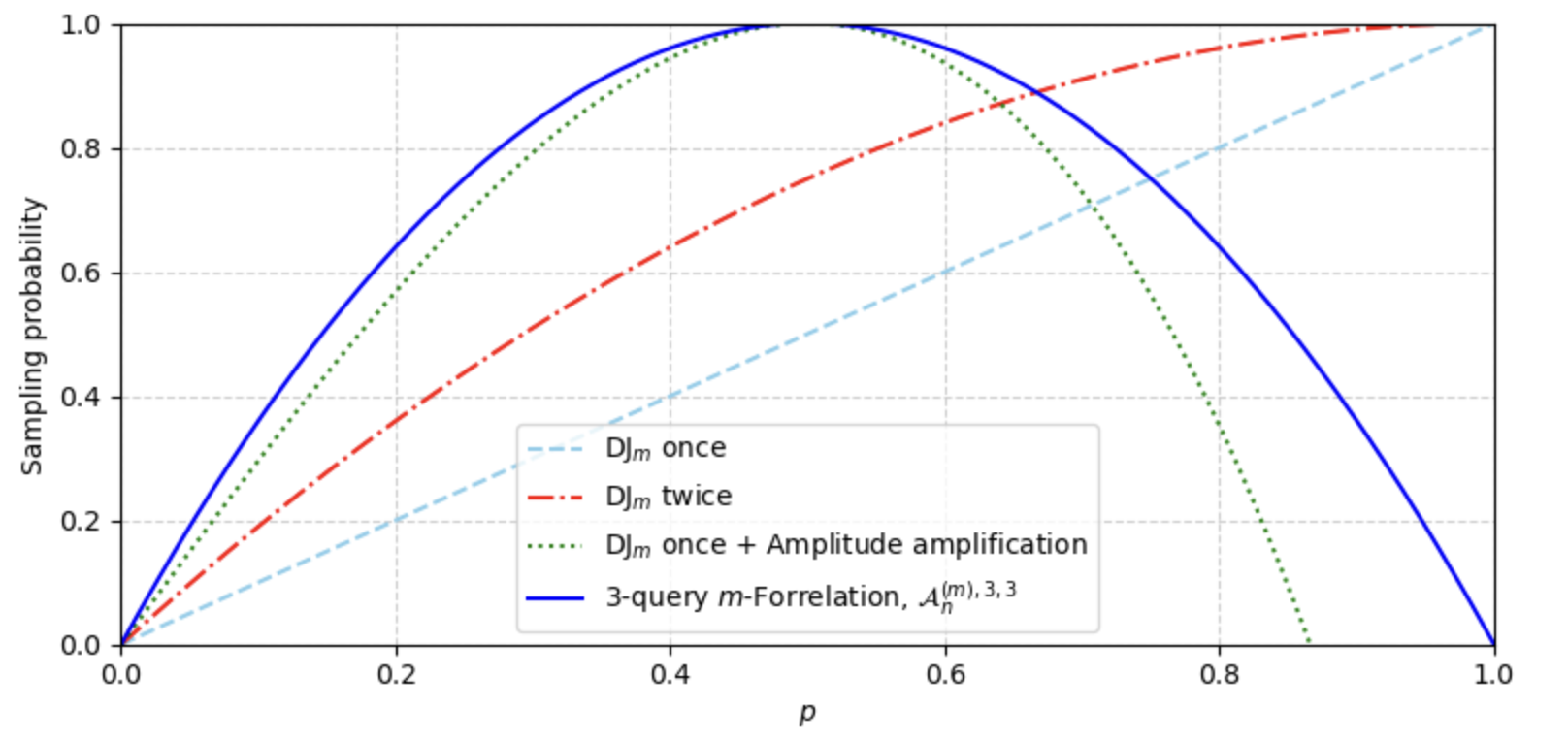}
		\caption{Sampling probabilities of $m$-Hadamard transform using different algorithms.}
		\label{fig:graph}
	\end{center}
\end{figure}

Next, we provide an efficient sampling of the $m$-crosscorrelation and $m$-autocorrelation, using the $m$-Forrelation algorithms, $A_n^{(m)3,3}$ and $A_n^{(m)2,3}$.

\begin{theorem}
\label{thm:sample-mcross}
Given oracle access to $f, g\in\mathcal{B}_n$, the $3$-query $m$-Forrelation algorithm, $A_n^{(m)3,3}$, estimates the $m$-crosscorrelation value of $f$ and $g$ at any given point $\var{y}\in \mathbb{F}_2^n$ by evaluating the probability of measuring the all-zero state, which is given by
$$\pr{\left(\ket{0^n}\right)}=\displaystyle2^{-2n}\left| C^{(m)}_{f, g}(\var{y})\right|^2.$$
Further, the $2$-query $m$-Forrelation algorithm, $A_n^{(m)2,3}$, estimates the real part of the $m$-crosscorrelation value of $f$ and $g$ at any given point $\var{y}\in\mathbb{F}_2^n$ by evaluating the probability of measuring the `driving qubit' in the $\ket{0}$ state, which is given by $$\pr{\left(\ket{0}\right)}=\displaystyle\frac{1}{2}\left(1+\Re{\left( 2^{-n}\zeta_m^{-wt(\var{y})}C^{(m)}_{f,g}(\var{y})\right)}\right).$$
\end{theorem}
\begin{proof}
	Suppose, $h(\xv)\in \mathcal{B}_n$ such that $h(\xv)={\xv\cdot\var{y}}$. Then, from Theorem~\ref{thm:m-cross-hadamard}, the $m$-crosscorrelation value at $\var{y}\in\mathbb{F}_2^n$ can be written as
	$$C^{(m)}_{f,g}(\var{y}) = \zeta_m^{wt(\var{y})}\displaystyle\sum_{\xv \in \mathbb{F}_2^n} \mathcal{H}^{(m)}_{f}(\xv)\overline{\mathcal{H}^{(m)}_{g}(\xv)} (-1)^{h(\xv)} = \zeta_m^{wt(\var{y})}2^n\cdot \Phi^{(m)}_{h,f,g}.$$
	This implies $\Phi^{(m)}_{h,f,g}=2^{-n}\zeta_m^{-wt(\var{y})}C^{(m)}_{f,g}(\var{y})$. Since the $m$-Forrelation values $\Phi^{(m)}_{h,f,g}$ can be estimated from  executing the algorithms $A_n^{(m)3,3}$ and $A_n^{(m)2,3}$, the rest of the proof follows.
\end{proof}

This gives us a constant query algorithm for sampling the $m$-crosscorrelation values of any two Boolean functions, $f,g\in\mathcal{B}_n$ at any given point, $\yv\in\mathbb{F}_2^n$. For $g=f$, we obtain a constant query sampling of $m$-autocorrelation, $C^{(m)}_{f}$ as an immediate corollary.
\begin{corollary}
\label{cor:sample-mauto}
Given oracle access to $f\in\mathcal{B}_n$, the $3$-query $m$-Forrelation algorithm, $A_n^{(m)3,3}$, estimates the $m$-autocorrelation value of $f$ at any given point $\var{y}\in\mathbb{F}_2^n$ by evaluating the probability of measuring the all-zero state, which is given by
$$\pr{\left(\ket{0^n}\right)}=\textstyle 2^{-2n} |C^{(m)}_{f}(\var{y})|^2.$$
Furthermore, the $2$-query $m$-Forrelation algorithm, $A_n^{(m)2,3}$, estimates the real part of the $m$-autocorrelation value of $f$ at any given point $\var{y}\in\mathbb{F}_2^n$ by evaluating the probability of measuring the `driving qubit' in the $\ket{0}$ state, which is given by $$\pr{\left(\ket{0}\right)}=\textstyle\frac{1}{2}\left[1+\Re{\left(2^{-n} \zeta_m^{-wt(\var{y})}C^{(m)}_{f}(\var{y})\right)}\right].$$
\end{corollary}

Moreover, following Theorem~\ref{thm:autobent}, we have another direct corollary as follows.
\begin{corollary}
	\label{cor:mautobent}
	Let $f\in\mathcal{B}_n$ be an $m$-bent function, and let $h\in\mathcal{B}_n$ be a linear Boolean function (i.e., either constant or balanced). Then the presence or absence of the all-zero state in the measurement outcome of the $3$-query algorithm $A_n^{(m)2,3}(g,f,f)$ determines whether $h$ is constant, i.e., $h(\xv)=0$, for all $\xv\in\mathbb{F}_2^n$, or balanced, respectively.
\end{corollary}
\begin{proof}
	From Corollary~\ref{cor:sample-mauto}, $\pr{\left(\ket{0^n}\right)}=\textstyle 2^{-2n}|C^{(m)}_{f}(\var{y})|^2.$ Since, $f\in\mathcal{B}_n$ be an $m$-bent, from Theorem~\ref{thm:autobent}, $\pr{\left(\ket{0^n}\right)} \neq 0$ if $\var{y} = 0^n$, i.e., $h(\var{x})=\var{x}\cdot \var{y} = 0$, and $\pr{\left(\ket{0^n}\right)} = 0$ if $\var{y} \neq 0^n$, i.e., $h(\var{x})=\var{x}\cdot \var{y}$, a balanced Boolean function.
\end{proof}

%\begin{theorem}
%Let $f_1,f_2\in\mathcal{B}_n$ with the promise that either $f_1=f_2$ or $f_2$ is a shifted version $f_1$, i.e., $f_2(\xv)=f_1(\xv\oplus \var{y})$. Then the $3$-fold nega-Forrelation algorithm, $\tilde{A}_n^{3,3}$ deterministically resolves which one it is using a single query to $f_1$ and $f_2$, each.
%\end{theorem}

Theorem~\ref{thm:sample-mcross} and Corollary~\ref{cor:sample-mauto} estimate the $m$-crosscorrelation and the $m$-autocorrelation values, respectively, at any specified point $\var{y}\in\mathbb{F}_2^n$, by choosing the linear function $h(\xv)={\xv\cdot \var{y}}$, which is dependent upon $\var{y}$. Next, we attempt to sample from the complete spectrum of the $m$-crosscorrelation (and thus the $m$-autocorrelation) by placing a superposition of all possible linear Boolean function, in place of $h$ (shown in Fig.~\ref{fig:algo}).
\begin{figure}[ht]
\begin{center}
\includegraphics[scale=1.1]{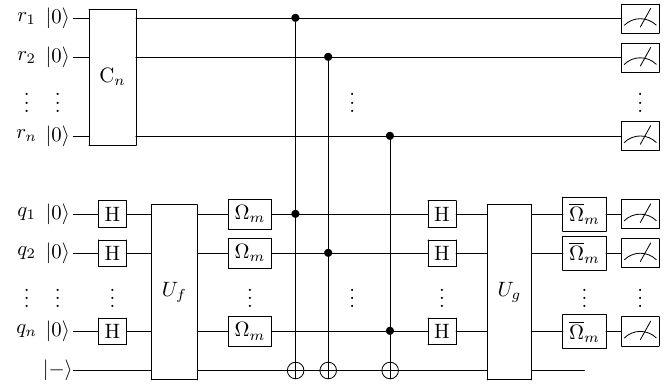}
\caption{Quantum circuit for sampling the complete $m$-crosscorrelation spectrum, denoted by Algorithm $\mathbb{A}^{(m)}(\mathrm{C}_n)$.}
\label{fig:algo}
\end{center}
\end{figure}

Let $\mathrm{C}_n$ be a $2^n \times 2^n$ unitary operator such that $\mathrm{C}_n \left(\ket{0^n}\right) = \sum_{\xv \in \mathbb{F}_2^n} \alpha_{\xv}\ket{\xv}$ with $\alpha_{\xv}\in\mathbb{C}$, for all $\xv\in\mathbb{F}_2^n$ satisfying $\sum_{\xv \in \mathbb{F}_2^n} \left|\alpha_{\xv}\right|^2 = 1$. Then, starting from the all-zero state, the algorithm $\mathbb{A}^{(m)}(\mathrm{C}_n)$ as in Fig.~\ref{fig:algo} has the pre-measurement state 
$$
\displaystyle \sum_{\var{y} \in \mathbb{F}_2^n} \alpha_{\var{y}}\ket{\var{y}} 
\left( \frac{C^{(m)}_{f,g}(\var{y})}{2^n}\ket{0^n} + \beta_{\var{y}}\ket{W_{\var{y}}} \right),
$$
where $\ket{W_\var{y}}$ is an $n$-qubit superposition state such that the amplitude of the state $\ket{0^n}$ is $0$. Theorem~\ref{thm:sample-fullcorr} shows that the specific choice of $\mathrm{C}_n$ helps in sampling the complete spectrum of $m$-crosscorrelation, as follows.

\begin{theorem}
\label{thm:sample-fullcorr}
Fixing $\mathrm{C}_n = \mathrm{H}^{\otimes n}$, the probability of observing $\ket{\var{y}}\ket{0^n}$ upon measuring the first $2n$ qubits from Algorithm ${\mathbb{A}}^{(m)}$ is given by $2^{-3n} \left|C^{(m)}_{f,g}(\var{y})\right|^2$ for all $\var{y} \in \mathbb{F}_2^n$. 
\end{theorem}
\begin{proof}
Fixing $\mathrm{C}_n= \mathrm{H}^{\otimes n}$, the pre-measurement state becomes 
$\displaystyle \sum_{\var{y} \in \mathbb{F}_2^n} 2^{-\frac{n}{2}}\ket{\var{y}} \left( \frac{C^{(m)}_{f,g}(\var{y})}{2^n}\ket{0^n} + \beta_{\var{y}}\ket{W_{\var{y}}}\right).$
Thus the probability of observing the state $\ket{\var{y}}\ket{0^n}$ upon measuring the first $2n$ many qubits is given by $2^{-3n} \left|C^{(m)}_{f,g}(\var{y})\right|^2$.
\end{proof}

The algorithm $\mathbb{A}^{(m)}$ is a generalization, where $m=1$ corresponds to sampling the complete spectrum of cross-correlation, as presented in~\cite[Algorithm 1]{amc}. For $m=4$, it enables the sampling of the complete spectrum of nega-crosscorrelation, as described in~\cite[Algorithm 1]{dcc}. Furthermore, for $m=2^k$, the algorithm can also be used for sampling the complete spectrum of $2^k$-crosscorrelation. In this context, we present the following corollary.

\begin{corollary}
\label{cor:mautobent-full}
Let $f\in\mathcal{B}_n$ be a $m$-bent function. Then, from executing the algorithm $\mathbb{A}^{(m)}(\mathrm{H}^{\otimes n})$ with $g=f$, the probability of observing the all-zero state, $\ket{0^n}\ket{0^n}$, upon measuring all the $2n$ qubits, is given by $2^{-n}$. Moreover, the probability of observing a state $\ket{\xv}\ket{0^n}$, where $\var{x}\neq 0^n$, is $0$.
\end{corollary}
The proof follows from Corollary~\ref{cor:mautobent}, using the fact that the $m$-autocorrelation value of an $m$-bent function is $0$ at any nonzero point.

A Dicke-state is an $n$-qubit quantum state with an equal superposition of all $\binom{n}{k}$ basis states with Hamming weight $k$. According to~\cite{dicke}, it is known that starting from $\ket{0^n}$, any Dicke state $\ket{D_k^n}$ can be deterministically prepared using $\mathcal{O}\left( n^2 \right)$ CNOT gates and $\mathcal{O}\left( n^2 \right)$ many single-qubit gates. Let $UD_k^n$ be the $2^n\times 2^n$ unitary operator that prepares the Dicke-state $\ket{D_k^n}$ of weight $k$ such that
$$\mathrm{UD}^{n}_{k} \left(\ket{0^n}\right) = \frac{1}{\sqrt{\binom{n}{k}}} \sum_{\xv: wt(\xv)=k} \ket{\xv}.$$
Consequently, if we replace $\mathrm{C}_n = \mathrm{UD}_k^n$ with $k < n$, the probability of observing $\ket{\var{y}}\ket{0^n}$ is given by $\frac{\left|C^{(m)}_{f,g}(\var{y})\right|^2}{\binom{n}{k}2^{2n}}$, where the Hamming weight of $\var{y}$ is $k$ and the probability of observing
$\ket{\var{y}}\ket{0^n}$ is zero if $wt(\var{y}) \neq k$. 
In this manner, one can sample the $m$-crosscorrelation (and hence the $m$-autocorrelation) values at all the points having an equal Hamming weight $k$.

\section{On affine transformation and shift of (generalized) bent functions}
\label{sec:cont3}
In this section, we study the affine transformations of (generalized) bent functions, beginning with an interesting observation on the affine properties of $m$-bent ones.
\begin{theorem}
    Let $f,g\in\mathcal{B}_n$ be two $m$-bent function such that $g(\xv) = f(A\xv\op \var{b}) \op \var{c}\cdot\xv \op d$, where $A$ is an $n$-dimensional orthogonal matrix, $\var{b}, \var{c}\in\mathbb{F}_2^n$, $d\in\mathbb{F}_2$. Then, for $\var{b}=0^n$, if $f$ is $m$-bent, so is $g$.
\end{theorem}
\begin{proof}
Combining Lemma~\ref{lem:m-forr-prop}~(a) and Lemma~\ref{lem:m-forr-prop}~(c), we obtain:
$$\mathcal{H}_g^{(m)}(\bm{\omega}) = (-1)^{d} 2^{-n/2} \sum_{\xv\in\mathbb{F}_2^n} (-1)^{f(\yv) \op A\left(\var{c}\oplus \bm{\omega}\right)\cdot \yv }\zeta_m^{wt(\yv)} = (-1)^{d}\mathcal{H}_f^{(m)}(A(\var{c}\op \bm{\omega})).$$
Therefore, $\left|\mathcal{H}_g^{(m)}(\bm{\omega}) \right| = \left|(-1)^{d}\mathcal{H}_f^{(m)}(A(\var{c}\op \bm{\omega})) \right| = \left| \mathcal{H}_f^{(m)}(A(\var{c}\op \bm{\omega})) \right|$. Hence, if $f$ is $m$-bent, then so is $g$.
%From Definition~\ref{def:m-forr}, we have
%$$\mathcal{H}_g^{(m)}(\bm{\omega}) = 2^{-n/2}\sum_{\var{x}}(-1)^{f(A\xv\op \var{b}) \op \var{c}\cdot\xv \op d \oplus \bm{\omega}\cdot\xv} \zeta_m^{wt(\xv)} = (-1)^{d}2^{-n/2} \sum_{\var{x}}(-1)^{f(A\xv\op \var{b}) \op (\var{c}\oplus \bm{\omega})\cdot\xv} \zeta_m^{wt(\xv)}.$$
%Since, $A$ is orthogonal, i.e., $A^TA = AA^T = I_n$, we set a new variable $\yv = A\xv \op \var{b}$, which implies $\xv = A^T(\yv \op \var{b})$. Further, $wt(\xv) = \xv^T I_n \xv$ implies that
%$$wt(A^T(\yv \op \var{b})) = \left( A^T(\yv \op \var{b}) \right)^T I_n \left(A^T(\yv \op \var{b})\right) = \left(\yv \op \var{b}\right)^T AI_n A^T \left( \yv \op \var{b}\right) = \left(\yv \op \var{b}\right)^T I_n \left( \yv \op \var{b}\right) = wt\left( \yv \op \var{b}\right)$$
%Therefore,
%\begin{align*}
    %\mathcal{H}_g^{(m)}(\bm{\omega}) = & (-1)^d 2^{-n/2} \sum_{\xv} (-1)^{f(\yv) \op (\var{c}\oplus \bm{\omega})\cdot A^T(\yv \op \var{b})}\zeta_m^{wt(\yv \op \var{b})}\\
    %= & (-1)^d 2^{-n/2} \sum_{\xv} (-1)^{f(\yv) \op A\left(\var{c}\oplus \bm{\omega}\right)\cdot (\yv \op \var{b})}\zeta_m^{wt(\yv \op \var{b})}\\
    %= & (-1)^{d\op A(\var{c}\op \bm{\omega})\cdot b} 2^{-n/2} \sum_{\xv} (-1)^{f(\yv) \op A\left(\var{c}\oplus \bm{\omega}\right)\cdot \yv }\zeta_m^{wt(\yv \op \var{b})}
%\end{align*}
\end{proof}

Notably, the similar results are already known for specific subcases such as standard bent functions, negabent functions, and $k$-bent functions. Consider two Boolean functions $f, g \in \mathcal{B}_n$ related by the affine transformation $g(\xv) = f(A\xv \op \var{b}) \op \var{c}\cdot \xv \op d$, where $A$ is an $n\times n$ orthogonal matrix over $\mathbb{F}_2$ (i.e., $AA^T = A^TA = I_n$), $\var{b}, \var{c} \in \mathbb{F}_2^n$, and $d\in\mathbb{F}_2$. Then, the Walsh-Hadamard transforms of $f$ and $g$ satisfy
\begin{align*}
    W_g(\bm{\omega}) = & 2^{-n/2}\sum_{\xv\in\mathbb{F}_2^n}(-1)^{f(A\xv \op \var{b}) \op \var{c}\cdot \xv \op d} (-1)^{\bm{\omega} \cdot \xv}\\
    = & (-1)^{d \op A(\bm{\omega} \op \var{c})\cdot b}2^{-n/2}\sum_{\yv\in\mathbb{F}_2^n}(-1)^{f(\yv)\op A(\bm{\omega} \op \var{c})\cdot \yv}\\
    = & (-1)^{d \op A(\bm{\omega} \op \var{c})\cdot b} W_f(A(\bm{\omega} \op \var{c})),
\end{align*}
which implies that if $g$ is bent, then $f$ is also bent~\cite[Theorem 2]{schmidt2008}. Similarly, for the nega-Hadamard transforms:
\begin{align*}
    N_g(\bm{\omega}) = & 2^{-n/2}\sum_{\xv\in\mathbb{F}_2^n}(-1)^{f(A\xv \op \var{b}) \op \var{c}\cdot \xv \op d} (-1)^{\bm{\omega} \cdot \xv} (i)^{wt(\xv)}\\
    = & (-1)^{d \op A(\var{c}\op \bm{\omega})\cdot b}2^{-n/2}\sum_{\yv\in\mathbb{F}_2^n}(-1)^{f(\yv)\op A(\var{c}\op \bm{\omega})\cdot \yv}(i)^{wt(\yv \op \var{b})}\\
    = & (-1)^{d \op A(\var{c}\op \bm{\omega})\cdot b} (i)^{wt(\var{b})}2^{-n/2}\sum_{\yv\in\mathbb{F}_2^n}(-1)^{f(\yv)\op A(\var{c}\op \bm{\omega})\cdot \yv \op \yv \cdot \var{b}}(i)^{wt(\yv)}\\
    = & (-1)^{d \op A(\var{c}\op \bm{\omega})\cdot b}(i)^{wt(\var{b})} N_f(A(\var{c}\op \bm{\omega}) \op \var{b}),
\end{align*}
which shows that if $g$ is negabent, then $f$ is also negabent~\cite[Theorem 3(d)]{pre}. A similar argument applies to the $2^k$-Hadamard transforms, yielding:
$$\mathcal{H}_g^{(2^k)}(\bm{\omega}) = (-1)^{d}\mathcal{H}_f^{(2^k)}(A(\var{c}\op \bm{\omega})).$$
Studying isomorphisms among Boolean functions is an interesting area of research from computational points of view (see~\cite{isomorphism-2011,isomorphism-2023} and the references therein). In this initiative, while we establish the theoretical associativity, it is not immediate how a quantum algorithm can efficiently recover the hidden parameters of the affine transformations given oracle access to the $m$-bent functions. Designing such algorithms remains an open and promising research direction for future work. On the other hand, for certain specific instances of bent and negabent functions, we discuss certain algorithms in the following subsections.

\subsection{Quantum algorithm for finding the hidden shift of bent functions}
In~\cite{dcc}, it was noted that the hidden shift finding algorithm for bent functions, proposed in~\cite[Theorem 4.1]{shift}, is structurally similar to the $2$-fold Forrelation algorithm. Given two bent functions $f,g\in \mathcal{B}_n$ satisfying $g(\xv) = f(\xv \op \var{b})$, the algorithm recovers the hidden shift $\var{b}$ using a single query to $g$ and the dual of $f$. Here, we extend this approach to analyze the affine properties of bent functions, in a more general setup.
\begin{theorem}
Let $f,g\in\mathcal{B}_n$ be bent functions such that $g(\xv) = f(\xv\op \var{b}) \op \var{c}\cdot\xv \op d$, where $\var{b}, \var{c}\in\mathbb{F}_2^n$ and $d\in\mathbb{F}_2$. Then, running the 2-fold Forrelation algorithm with $U_g$ as the first oracle and the dual of $f$, $U_{\tilde{f}}$, as the second (see Fig.~\textup{\ref{fig:affine-bent}}), the state before measurement is given by $(-1)^{\var{b}\cdot \var{c} \op d}2^{-3n/2}\sum_{\var{y},\var{z}\in\mathbb{F}_2^n}(-1)^{\tilde{f}(\yv \op \var{c})\op \tilde{f}(\yv) \op \yv\cdot (\var{z}\op \var{b})}\ket{\var{z}}$, and the probability of observing any state $\var{z} \in \mathbb{F}_2^n$ is
$$\pr{\left(\var{z}\right)} = 2^{-3n}\big|\sum_{\yv\in\mathbb{F}_2^n}(-1)^{\tilde{f}(\yv \op \var{c})\op \tilde{f}(\yv) \op \yv\cdot (\var{z}\op \var{b})}\big|^2.$$
\end{theorem}
\begin{proof}
We follow the 2-fold Forrelation circuit from~\cite{shift}, illustrated in Fig.~\ref{fig:affine-bent}.
\begin{figure}[ht]
    \centering
    \includegraphics[scale=1.2]{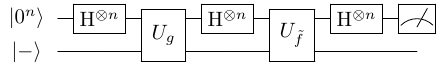}
    \caption{Quantum algorithm to study the affine properties of $f$}
    \label{fig:affine-bent}
\end{figure}
The starting state $\ket{0^n}$ evolves as:
\begin{align*}
    \ket{0^n} \xrightarrow{H^n} 2^{-n/2}\sum_{\xv\in\mathbb{F}_2^n}\ket{\xv}
    \xrightarrow{U_{g}} &\, 2^{-n/2}\sum_{\xv\in\mathbb{F}_2^n}(-1)^{f(\xv\op \var{b}) \op \var{c}\cdot\xv \op d}\ket{x}\\
    \xrightarrow{H^n} &\, (-1)^d2^{-n}\sum_{\xv,\yv \in \mathbb{F}_2^n}(-1)^{f(\xv\op \var{b}) \op \var{c}\cdot\xv \op \xv \cdot \yv}\ket{\yv}\\
    = &\, (-1)^{\var{b}\cdot \var{c} \op d}2^{-n}\sum_{\xv,\yv\in\mathbb{F}_2^n}(-1)^{\var{b}\cdot\yv}(-1)^{f(\xv\op \var{b}) \op (\xv \op \var{b})\cdot (\yv\op \var{c})}\ket{\yv}\\
    = &\, (-1)^{\var{b}\cdot \var{c} \op d}2^{-n}\sum_{\yv\in\mathbb{F}_2^n}(-1)^{\var{b}\cdot\yv}(-1)^{\tilde{f}(\yv \op \var{c})}\ket{\yv}\\
    \xrightarrow{U_{\tilde{f}}} &\, (-1)^{\var{b}\cdot \var{c} \op d}2^{-n}\sum_{\yv\in\mathbb{F}_2^n}(-1)^{\var{b}\cdot\yv}(-1)^{\tilde{f}(\yv \op \var{c})\op \tilde{f}(\yv)}\ket{\yv}\\
    \xrightarrow{H^n} &\, (-1)^{\var{b}\cdot \var{c} \op d}2^{-3n/2}\sum_{\yv,\var{z}\in\mathbb{F}_2^n}(-1)^{\var{b}\cdot\yv}(-1)^{\tilde{f}(\yv \op \var{c})\op \tilde{f}(\yv)}(-1)^{\yv \cdot \var{z}}\ket{\var{z}}.
\end{align*}
Hence, the final amplitude of any basis state $\var{z} \in \mathbb{F}_2^n$ is $(-1)^{\var{b}\cdot \var{c} \op d}2^{-3n/2}\sum_{\yv\in\mathbb{F}_2^n}(-1)^{\tilde{f}(\yv \op \var{c})\op \tilde{f}(\yv)}(-1)^{\yv \cdot (\var{b} \op \var{z})}$, and the corresponding measurement probability is given by $2^{-3n}\big|\sum_{\yv\in\mathbb{F}_2^n}(-1)^{\tilde{f}(\yv \op \var{c})\op \tilde{f}(\yv) \op \yv\cdot (\var{b}\op \var{z})}\big|^2.$
%The claim of our theorem is shown.
\end{proof}

Clearly, from the above theorem, the value of $d\in\mathbb{F}_2$ cannot be estimated. Furthermore, if $\var{c}=0^n$, the state $\var{b}$ is observed with probability 1, as shown in~\cite{shift}. When $\var{b}=0^n$, the state $\ket{0^n}$ is never observed. In general, with high probability, the observed state corresponds to the best affine approximation of the balanced Boolean function $\tilde{f}(\yv \op \var{c})\op \tilde{f}(\yv)$. Moreover, if $\var{b}\neq 0^n \neq \var{c}$, then the state $\var{b}$ is never observed. Instead, the most probable outcome is the state $\var{z}$, where the linear function $(\var{b}\op \var{z})\cdot \yv$ best approximates $\tilde{f}(\yv \op \var{c})\op \tilde{f}(\yv)$.

Summary of observations:
\begin{enumerate}
    \item If the algorithm yields only the state $\ket{0^n}$, then either $|f(\xv)|=|g(\xv)|$ for all $\xv \in \mathbb{F}_2^n$, or $\var{b}\neq 0^n\neq \var{c}$, with $\var{b}\in\mathbb{F}_2^n$ being the coefficient of the best linear approximation to $\tilde{f}(\yv \op \var{c})\op \tilde{f}(\yv)$.
    \item If a fixed nonzero state $\ket{\var{z}}$ is observed with probability 1, then either $\var{c}=0^n$ and the observed state is $\var{b}\in\mathbb{F}_2^n$, or $\var{c}\neq 0^n$ and $\tilde{f}(\yv \op \var{c})\op \tilde{f}(\yv) = (\var{b}\op \var{z})\cdot y$.
    \item If all the states are observed except for a fixed nonzero state, then $\var{c}\neq 0^n$, and the missing state corresponds to the hidden shift $\var{b} \in \mathbb{F}_2^n$.
\end{enumerate}
\begin{remark}
    Note that here we do not introduce any structural modifications to the algorithm; rather, we present an alternative analysis based on the affine transformations of the underlying Boolean functions.
\end{remark}

\subsection{Quantum algorithm for finding the hidden shift of negabent functions}

In this subsection, we study the shift properties of negabent functions and show that the existing hidden shift algorithm for bent functions can be adapted to recover the hidden shift of negabent functions.

We begin with identifying an error in a prior work. In~\cite[Theorem 5]{dcc}, it was wrongly claimed that if $f,g\in\mathcal{B}_n$ are both negabent, then there exists no $\var{u}\in\mathbb{F}_2^n$ such that $g(\xv)=f(\xv\op\var{u})$ for all $\xv\in\mathbb{F}_2^n$. An immediate counter example can be provided by any affine Boolean function $f(\var{x}) = \var{a}\cdot \var{x} \op c$, with $\var{a}\in\mathbb{F}_2^n$, $c\in\mathbb{F}_2$. Defining $g(\xv)=f(\xv\op \var{u})$, we obtain $g(\xv) = \var{a}\cdot \xv \op c'$ where $c' = \var{a}\cdot \var{u} \op c$ is a constant. Thus, $g$ is also affine and therefore negabent, satisfying $g(\xv) = f(\xv \op \var{u})$ and contradicting the claim. We present another counterexample to~\cite[Theorem 5]{dcc} by providing a non-affine, non-bent 6-variable negabent function, $f(x_1,x_2,x_3,x_4,x_5,x_6) = x_1x_3 \op x_1x_4$, where for a non-trivial shift, $\var{u} = 100001$, we obtain
$$f(\xv \op \var{u}) = f(x_1\op 1,x_2,x_3,x_4,x_5,x_6\op 1) = (x_1\op 1)x_3 \op (x_1\op 1)x_4 = x_1x_3 \op x_1x_4 \op x_3 \op x_4,$$
which is again a non-affine, non-bent, negabent function.

Additionally, let $f,g\in\mathcal{B}_n$ be such that $g(\xv) = f(\xv \op \var{u})$ for some $\var{u}\in\mathbb{F}_2^n$. If both $f$ and $g$ are bent as well as negabent, then the same $\var{u}$ serves as the shift for the negabent functions $f$ and $g$. However, $\var{u}$ cannot be a period of the corresponding negabent functions $f\op s_2$ and $g\op s_2$, as identified in~\cite[Definition 6]{dcc}. 
Specifically, if $(g\op s_2)(\xv) = (f\op s_2)(\xv \op \var{u})$, then it follows that $g(\xv) \op s_2(\xv) = f(\xv \op \var{u}) \op s_2(\xv) \op \mathbb{L}_\var{u}(\xv) \op s_2(\var{u})$. Canceling common terms yields $\mathbb{L}_\var{u}(\xv) \op s_2(\var{u}) = 0$. Thus, the linear function $\mathbb{L}_\var{u}(\xv) = \bigoplus_{i=1}^nx_i\left(\oplus_{j=1,j\neq i}^n u_i\right)$ must be a constant, which implies $\var{u}=0^n$.
We used above the known result of~\cite{parker2007}: if $f$ is bent, $f\op s_2$ is negabent, and vice versa, where $s_2$ denotes the quadratic symmetric Boolean function, which is itself a bent function.

On the other hand, since $\mathbb{L}_\var{u}(\xv \op \var{u}) = \mathbb{L}_\var{u}(\xv)$, we obtain the following immediate results.
\begin{lemma}
\label{lem:nega-shift}
    Let $f,g\in\mathcal{B}_n$ be two bent functions such that $g(\xv) = f(\xv \op \var{u)}$. Then $\var{u}$ is a period of the negabent functions $g\op s_2$ and $\left( f\op s_2 \op \mathbb{L}_{\var{u}} \op s_2(\var{u})\right)$, satisfying $(g\op s_2)(\xv) = \left( f\op s_2 \op \mathbb{L}_{\var{u}} \op s_2(\var{u})\right)(\xv \op \var{u}).$ Moreover, their duals $\Tilde{f}$ and $\Tilde{g}$ differ by the linear function $\var{u}\cdot \xv$, i.e., $\Tilde{g}(\xv) = \Tilde{f}(\xv) \op \var{u}\cdot \xv$.
\end{lemma}
\begin{proof}
We begin with $\left( f\op s_2 \op \mathbb{L}_{\var{u}} \op s_2(\var{u})\right)$:
\begin{align*}
    \left( f\op s_2 \op \mathbb{L}_{\var{u}} \op s_2(\var{u})\right)(\xv \op \var{u}) = & f(\xv \op \var{u}) \op s_2 (\xv \op \var{u}) \op \mathbb{L}_{\var{u}}(\xv \op \var{u}) \op s_2(\var{u})\\
    = & g(\xv) \op s_2(\xv) \op \mathbb{L}_{\var{u}}(\xv) \op \mathbb{L}_{\var{u}}(\xv)\\
    = & (g \op s_2)(\xv).
\end{align*}
Furthermore, $(-1)^{\tilde{g}(\xv)} = 2^{n/2}W_g(\xv) = 2^{n/2}W_f(\xv)(-1)^{\var{u}\cdot \xv} = (-1)^{\tilde{f}(\xv)}(-1)^{\var{u}\cdot \xv} = (-1)^{\tilde{f}(\xv) \op \var{u}\cdot \xv}$. Hence, $\tilde{f}$ and $\tilde{g}$ differ by the linear function $\var{u}\cdot \xv$.
\end{proof}
%From~\cite{parker2007} it is known that if $f$ is bent, $f\op s_2$ is negabent, and vice versa, where $s_2$ denotes the quadratic symmetric Boolean function, which is itself a bent. Consequently, if $f$ is bent negabent, then so is $f\op s_2$. Therefore, every bent function has a corresponding negabent function. Further, if $f$ is said to be bent, but not negabent, then $f\op s_2$ is negabent, but not bent. This can be shown by contradiction, as if $f\op s_2$ was bent too, then $f\op s_2 \op s_2 = f$ would become negabent, leading to contradiction.

Given two bent functions $f,g\in\mathcal{B}_n$ satisfying $g(\xv) = f(\xv \op \var{u})$,~\cite[Theorem 4.2]{shift} also proposed a quantum algorithm to compute the hidden shift $\var{u}\in\mathbb{F}_2^n$ with constant success probability, by making $\mathcal{O}(n)$ many queries to $f$ and $g$. Upon measurement, the algorithm yields the states $\ket{b,\var{z}}$ which are orthogonal to the hidden shift $1,\var{u}$, i.e., $(b,\var{z})\cdot (1,\var{u}) = b\op \var{u}\cdot \var{z} = 0$. From the observed states, $\var{u}$ can be efficiently recovered via Gaussian elimination. In this context, we observe that the same algorithm can be adapted to recover the hidden shift when $f$ and $g$ are negabent functions.

To demonstrate this, we apply the hidden shift finding algorithm to the non-trivial 6-variable negabent function $f(x_1,x_2,x_3,x_4,x_5,x_6) = x_1x_3\op x_1x_4$, which is not bent, along with its shifted version $g(\xv) = f(\xv \op 100001) = x_1x_3 \op x_1x_4 \op x_3 \op x_4$ (see Fig.~\ref{fig:qcirc}). The observed output states $\ket{0000000}$, $\ket{1100000}$, $\ket{0001100}$, and $\ket{1101100}$ appear with equal probability (see Fig.~\ref{fig:qhist}), indicating that the hidden shift $\var{u} = u_1u_2u_3u_4u_5u_6$ satisfies the constraints $u_1 = 1$, $u_3=u_4$, and $u_2,u_3,u_5,u_6 \in \mathbb{F}_2$. In fact, there are exactly $2^4$ such valid shifts $\var{u}\in\mathbb{F}_2^6$ satisfying $g(\xv) = f(\xv \op \var{u})$ where $f(\xv) = x_1x_3\op x_1x_4$ and $g(\xv) = x_1x_3 \op x_1x_4 \op x_3 \op x_4$, under the conditions $u_1 = 1$ and $u_3 = u_4$.
\begin{figure}[htbp]
\centering
\begin{subfigure}{.64\textwidth}
  \centering
  \includegraphics[width=1\linewidth]{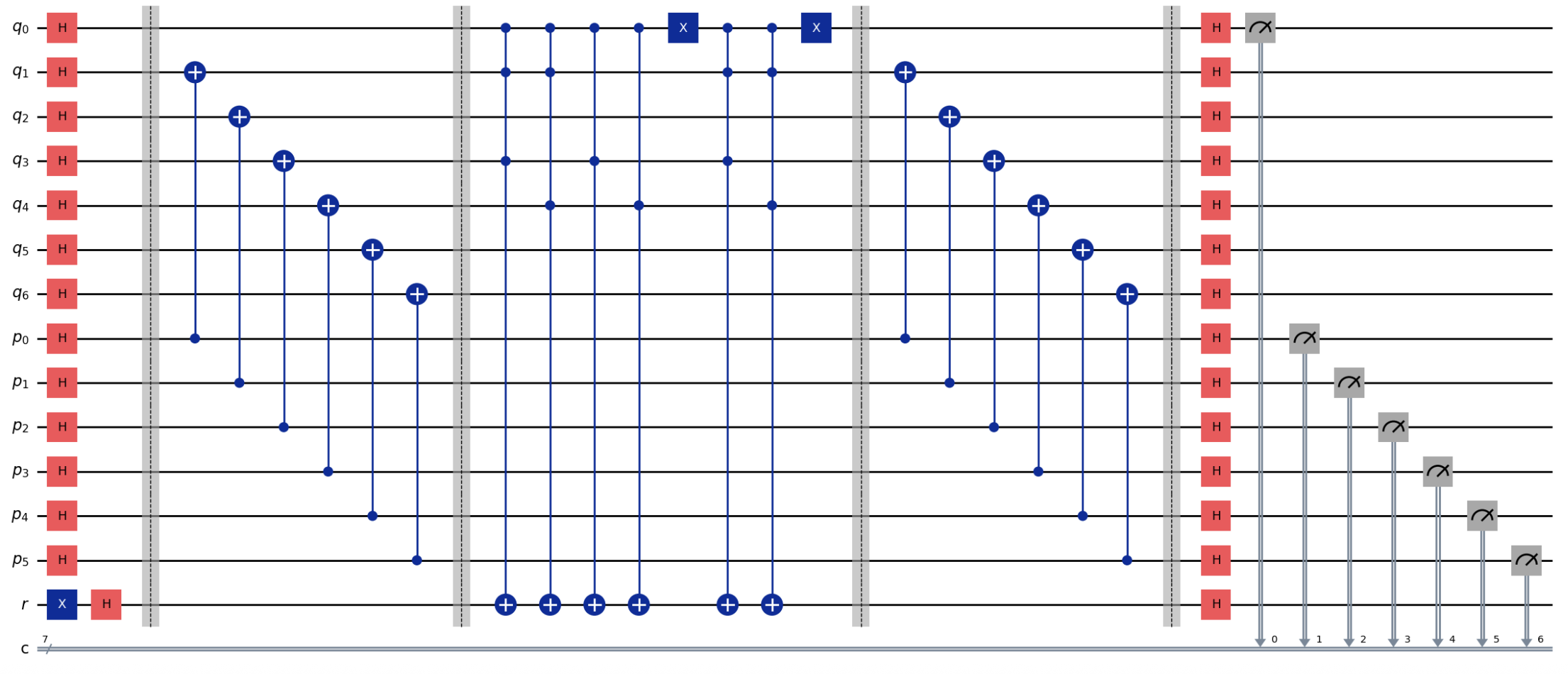}
  \caption{Quantum circuit for finding the hidden shift.}
  \label{fig:qcirc}
\end{subfigure}%
\begin{subfigure}{.33\textwidth}
  \centering
  \includegraphics[width=1\linewidth]{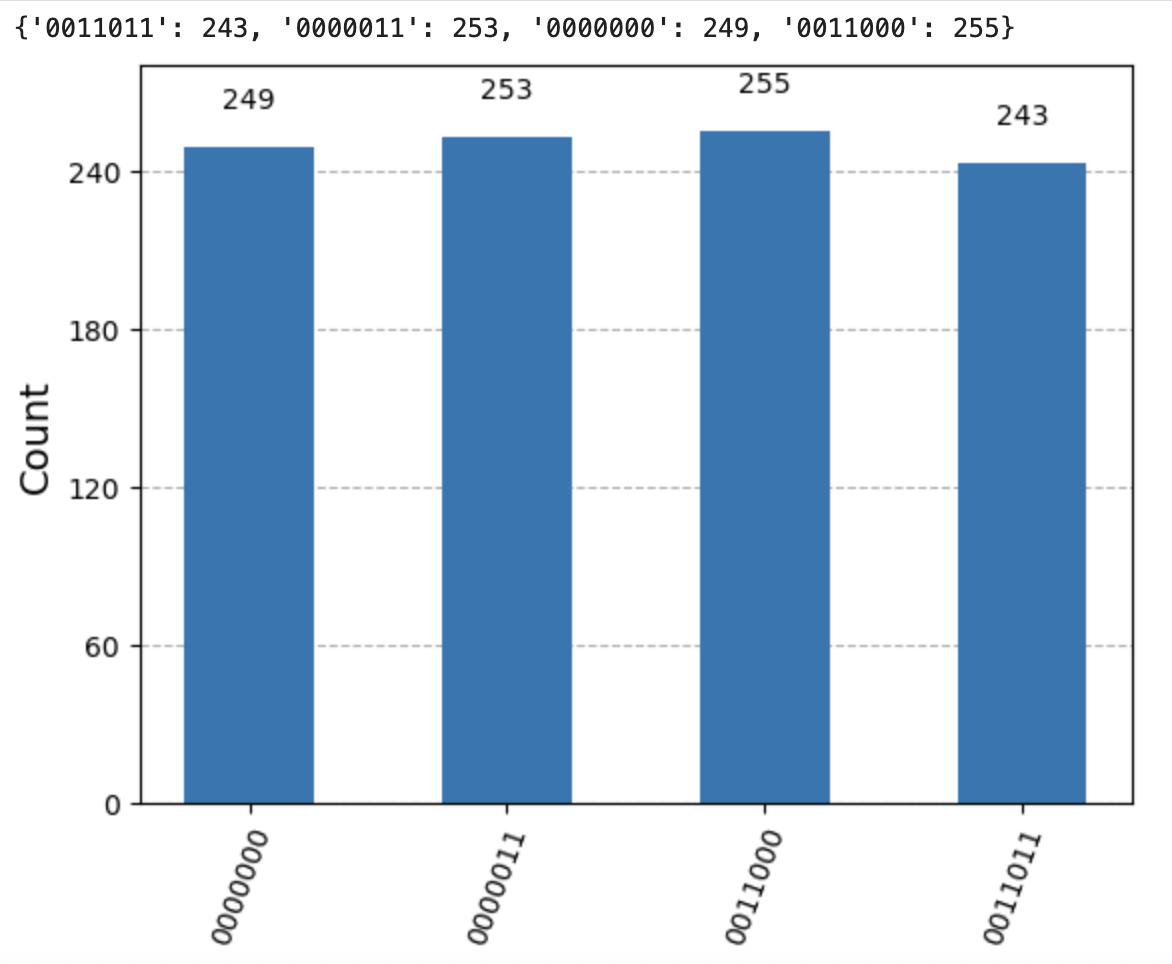}
  \caption{Corresponding histogram.}
  \label{fig:qhist}
\end{subfigure}
\caption{Quantum algorithm for computing the hidden shift from negabent functions, $f(\xv) = x_1x_3\op x_1x_4$ and $g(\xv) =f(\xv \op 100001) = x_1x_3 \op x_1x_4 \op x_3 \op x_4$.}
\label{fig:shift-negabent}
\end{figure}
\begin{remark}
    The hidden shift finding algorithm for negabent functions studied here is fundamentally and structurally different from the `nega-shift' recovery approach proposed in~\textup{\cite[Theorem 6]{dcc}}. In~\textup{\cite{dcc}}, the underlying negabent functions were of a specific form (as in \textup{Lemma~\ref{lem:nega-shift}}), and the functions were first transformed into their corresponding bent counterparts. The hidden shift was then recovered using a single-query algorithm, analogous to the 2-fold Forrelation technique. In contrast, the present approach imposes no particular structure on the underlying Boolean functions. Instead, we construct a large class of quantum states orthogonal to the hidden shift and recover the desired shift via classical post-processing, such as Gaussian elimination.
\end{remark}

\section{Conclusion}
\label{sec:con}
%In this section, we summarize the key results of this paper and conclude with further research possibilities in this direction.

In this paper, we generalized various cryptographically significant spectra of Boolean functions, including the Walsh-Hadamard, cross-correlation, and autocorrelation spectra, extending previous formulations to any $m\in\mathbb{N}$. We demonstrated that existing variants—such as the standard version, the nega variant, and the $2^k$-variant—are special cases of this more general framework. Additionally, we identified a previously unexamined class of real Hadamard transforms that lies between the Walsh-Hadamard and nega-Hadamard transformations, filling a gap in the existing literature. Furthermore, we introduced the most generalized version of the Deutsch-Jozsa algorithm, which extends both the standard Deutsch-Jozsa and its prior extended version, thereby encompassing them as special cases. We established that the pre-measurement state in the generalized Deutsch-Jozsa algorithm corresponds to a superposition of the normalized $m$-Hadamard transforms evaluated at all points. In addition, we extended the Forrelation formulation to $m$-Forrelation, and presented new quantum algorithms for estimating these newly defined generalized spectra. We believe these techniques may be of use to understand the interactions among Boolean functions, various spectra and relevant quantum algorithms.  

%We also generalized the bent property of Boolean functions based upon the flat-ness of their $m$-Hadamard transform values and proposed a couple of promise problems in this context. However, we did not delve deeply into the implications of the $m$-bent property, its testing algorithms, or its broader cryptographic significance, leaving these as open directions for future research.

\end{document}